\documentclass[
  aps,
  prx,
  twocolumn,
  superscriptaddress,
  nofootinbib,
  longbibliography,
  floatfix
]{revtex4-2}

\usepackage{amsmath}
\usepackage{amssymb}
\usepackage{amsthm}
\usepackage{bm}
\usepackage{braket}
\usepackage{graphicx}
\usepackage{xcolor}
\usepackage{hyperref}
\usepackage{quantikz}
\usepackage[caption=false]{subfig} 

\theoremstyle{plain}
\newtheorem{theorem}{Theorem}
\newtheorem{proposition}[theorem]{Proposition}
\newtheorem{lemma}[theorem]{Lemma}
\newtheorem{corollary}[theorem]{Corollary}
\newtheorem{definition}[theorem]{Definition}

\theoremstyle{definition}
\newtheorem{remark}[theorem]{Remark}

\newcommand{\bx}{\bm{x}}
\newcommand{\bz}{\bm{z}}
\newcommand{\ba}{\bm{a}}

\newcommand{\bv}{\bm{v}}

\newcommand{\bphi}{\bm{\phi}}
\newcommand{\onenorm}[1]{\lVert #1 \rVert_{1}}
\newcommand{\twonorm}[1]{\lVert #1 \rVert_{2}}
\newcommand{\supnorm}[1]{\lVert #1 \rVert_{\infty}}
\newcommand{\Ex}{\mathbb{E}}

\newcommand{\inner}[2]{\langle #1, #2 \rangle}
\newcommand{\I}{\mathrm{i}}

\newcommand{\BigO}{\mathcal{O}}
\newcommand{\bbx}{\bm{\bar{x}}}

\DeclareMathOperator{\Span}{\mathrm{span}}

\DeclareMathOperator{\supp}{supp}
\DeclareMathOperator*{\Exp}{\mathbb{E}}

\begin{document}

\title{Custom Penalty QAOA via Walsh--Hadamard Transform}

\author{Xin Wei LEE}
\email{xwlee@smu.edu.sg}
\affiliation{School of Computing and Information Systems, Singapore Management University, 80 Stamford Road, Singapore 178902.}

\author{Siong Thye GOH}
\email{gohst2@a-star.edu.sg}
\affiliation{School of Computing and Information Systems, Singapore Management University, 80 Stamford Road, Singapore 178902.}
\affiliation{Institute of Advanced Intelligence and Computing (IAIC),  Agency for Science, Technology and Research (A*STAR), 1 Fusionopolis Way, Singapore 138632.}

\author{Hoong Chuin LAU}
\email{hclau@smu.edu.sg}
\affiliation{School of Computing and Information Systems, Singapore Management University, 80 Stamford Road, Singapore 178902.}
\affiliation{Institute of Advanced Intelligence and Computing (IAIC),  Agency for Science, Technology and Research (A*STAR), 1 Fusionopolis Way, Singapore 138632.}

\date{\today}

\begin{abstract}
Inequality constraints in combinatorial optimization are conventionally handled by introducing slack variables and quadratic penalties, which increase the qubit count and distort the optimization landscape. Slack-free formulations instead apply a nonlinear custom penalty function directly to the constraint, which are difficult to be implemented in quantum approximate optimization algorithm (QAOA). We show that the Walsh--Hadamard transform (WHT) removes this obstruction for any penalty function: a penalty applied to a linear constraint is exactly a weighted sum of Pauli-$Z$ strings. We propose this framework as the Walsh custom penalty QAOA (WCP-QAOA). Truncating the expansion at Walsh degree $k$ costs $\mathcal{O}(n^k)$ Pauli terms and leaves the qubit count untouched. The truncation is the optimal degree-$k$ approximation of the penalty in $L^2$, and we give a condition on the discarded Walsh coefficients under which the truncated Hamiltonian still ranks every feasible solution below every infeasible one. Since evaluating a Walsh coefficient from its definition requires all $2^n$ penalty values, we give an algorithm that returns any coefficient in sub-exponential time, so that the degree-$k$ Hamiltonian is obtained in time polynomial in $n$ for fixed $k$. We show that for the exponential penalty, there exists a closed product form of WHT, and hence a closed form expression to calculate the Walsh coefficients. We identify structural shortcuts that shorten the exact spectrum. We also tested WCP-QAOA in solving a 6-qubit knapsack problem.
\end{abstract}

\maketitle

\section{Introduction}

Constrained combinatorial optimization is one of the emerging applications for near-term quantum computing~\cite{ising-np,challenges-qo}.
The two popular variational algorithms, the quantum approximate optimization algorithm (QAOA)~\cite{farhi2014quantum} and the 
variational quantum eigensolver (VQE)~\cite{vqe-first,vqe-review}, both operate on an unconstrained problem Hamiltonian in their typical form,
so a constrained problem must first be converted into an unconstrained one. The standard conversion is the quadratic unconstrained binary optimization (QUBO) reformulation,
in which each violated constraint contributes an additive energy penalty. Equality constraints are squared directly. 
Inequality constraints, however, must first be converted into equalities by introducing slack variables, and each slack variable must itself be encoded in binary.

The slack variables are the source of two distinct problems. The first is a resource problem: each inequality constraint contributes $\lceil \log_2 s_{\max} \rceil$ additional qubits,
where $s_{\max}$ is the largest slack value the constraint can take, and qubits are scarce resources in near-term quantum hardware.
The second is a landscape problem. For a fixed feasible bit-string $\bx$, only one value of the slack register satisfies the derived equality,
every other value produces a nonzero penalty even though $\bx$ itself is feasible. The optimizer therefore has to locate a joint assignment of problem and slack variables,
and low QUBO energy no longer faithfully corresponds to a good solution of the original problem~\cite{in-constraint,assist-cvrp}.

This has motivated a family of slack-free encodings. In the context of quantum annealing, several works reformulated the partition function so that only constraint-satisfying
configurations contribute, with each inequality enforced by a Heaviside step function rather than by a large quadratic penalty~\cite{ohzeki-inequality,ohzeki-subgradient}.
On gate-based devices, Golden \emph{et al.}\ replaced the QAOA phase separator by a threshold function that flags solutions above a given objective value~\cite{th-gm-qaoa}.
Monta{\~n}ez-Barrera \emph{et al.}\ proposed unbalanced penalization, which uses a truncated Taylor expansion to approximate the exponential penalty of inequality 
constraints~\cite{ub-penalty,ub-dwave}.

Our own previous work pursued a different route. In Ref.~\cite{step-qce} we introduced the custom penalty VQE (CP-VQE),
in which the inequality constraint is penalized by an arbitrary indicator function $\xi$ applied directly to the constraint value,
with the exponential and the Heaviside step function as the two concrete choices. Because the constraint Hamiltonian is diagonal,
$\xi$ acts entrywise on its eigenvalues. The loss function never has to be expressed as a polynomial because we used a problem-agnostic ansatz.
The penalty expectation are evaluated classically with $\BigO(2^m)$ time, where $m$ is the number of variables in the constraint.
In Ref.~\cite{cvar-cp} we removed the exponential cost of evaluating the penalty expectation exactly by replacing it with a finite-sampling estimator,
and combined the custom penalty with conditional-value-at-risk (CVaR) training~\cite{cvar,opt-var},
which emphasizes the low-energy tail of the sampled distribution, on the multidimensional knapsack problem (MDKP).

Both of these works use VQE with a hardware-efficient ansatz (HEA)~\cite{hea}, and for a good reason: the nonlinear custom penalty cannot be directly placed inside a QAOA phase separator.
The QAOA trial state is built from alternating unitaries $e^{-\I\gamma H_C}$ and $e^{-\I\beta H_M}$, and $H_C$ appears in an exponent.
One can of course compute $h_j(\bx)$ into an ancilla register with a quantum adder~\cite{addition-qc,ripple-carry-circuit} and a comparator and phase-kick the result~\cite{if-qaoa},
but that reintroduces exactly the qubit overhead that the slack-free formulation was meant to remove, and it increases the gate count and depth substantially.
The approach in \cite{if-qaoa} requires a Quantum Phase Estimation (QPE) subroutine where the $m$ phase-addition steps and the QFT must be executed in sequence and then fully uncomputed,
resulting in a circuit depth that scales as $\mathcal{O}(m+M)$ with large constants  where $M$ is the number of ancillary qubits.
For a Knapsack instance of $m=20$ variables in \cite{if-qaoa}, this amounts to $94$ layers versus $28$ for the penalty-QUBO approach.
Detailed resource estimates are provided in Appendix C of~\cite{if-qaoa}.

Giving up QAOA is not attractive either. QAOA is the Trotterized form of adiabatic quantum computation, so its structure inherits the adiabatic guarantee in the large-depth limit,
and it possesses a body of structure that hardware-efficient ansatz do not: periodicity and symmetry of the expectation landscape~\cite{bilinear,Lotshaw_2021},
and strong regularities in the optimal parameters across angle index and circuit depth~\cite{leo2020,qaoa-mvc}.
These regularities are the basis of a large family of initialization heuristics: linear ramp schedules~\cite{towards-lr}, \textsc{INTERP} and \textsc{FOURIER}~\cite{leo2020},
parameter fixing and the bilinear strategy~\cite{pf-qaoa,bilinear}, and annealing-inspired (TQA) initialization~\cite{tqa2021}, 
which make deep QAOA circuits trainable in practice.
QAOA also naturally composes with feasibility-preserving mixers such as the $XY$ mixer~\cite{xy-mixer,constraint-preserving-qaoa} and the Grover mixer~\cite{grover-mixer},
which act on the mixing operator rather than on the cost Hamiltonian and are therefore complementary to a penalty-based treatment of inequality constraints.

This work therefore asks a concrete question: \emph{can we express the nonlinear penalty terms as a sum of Pauli observables?}
The tool is the Walsh--Hadamard transform (WHT). Any real function of $n$ bits has a unique multilinear expansion in the Walsh basis,
and Hadfield showed that this expansion is literally the Pauli-$Z$ representation of the corresponding diagonal Hamiltonian~\cite{boolean-ham}.
The same correspondence underlies the ancilla-free construction of diagonal unitaries from the Walsh series by Welch \textit{et al.}~\cite{welch-eff-qc}.
In the midst of working on this manuscript, we also discovered that Lee \emph{et al.} had used similar approach to formulate nonlinear penalties with WHT, projecting it 
to the space of quadratic loss function, i.e. a QUBO, to solve the problem on quantum annealers~\cite{loong-wht}.
On the other hand, this work focuses more on QAOA and the properties of a truncated WHT applied to linear constraints.
QAOA, of course, has the capability of handling higher-order terms, since it does not restrict the order of the loss function.
The truncation degree $k$ of the Walsh series is therefore a parameter to be decided by the user. 
A custom penalty $\xi(\ba^{\top}\bx - b)$ is such a function, so it has an exact Pauli-$Z$ form, and the phase separator can be implemented exactly,
with no ancilla and with no Trotter error, as all the phase-separation terms commute.
The difficulty is that the expansion has $2^n$ terms in general, and computing even one coefficient from the definition requires all $2^n$ function values.

Our contributions address both the size of the spectrum and the cost of obtaining it:

\begin{enumerate}
    \item We give the spectral form of the custom penalty phase separator: the penalty Hamiltonian is exactly a sum of commuting Pauli-$Z$ strings,
    so $e^{-i\gamma H_C}$ factorizes into $\mathrm{CNOT}$ ladders and $Z$ rotations with no ancilla and no Trotter error.
    Truncating at Walsh degree $k$ costs $O(n^k)$ Pauli terms with the qubit count fixed at $n$.
    We prove that the truncation preserves the feasible-infeasible ordering whenever the discarded Walsh weight is small compared with the penalty gap,
    which for the step penalty means below one half.
    \item We give an algorithm that computes any single Walsh coefficient of a penalty $\xi(\ba^{\top}\bx - b)$ in $O(n\onenorm{\ba})$ time and $O(\onenorm{\ba})$ memory,
    without enumerating the $2^n$ bit-strings and without sampling.
    For fixed $k$ the full degree-$k$ Hamiltonian follows in time polynomial in $n$, and the same identity yields exactly the retained spectral power.
    \item We show that the truncated Walsh spectrum is optimal in approximating the original penalty function.
    The Walsh truncation is never worse in $L^2$ than any polynomial surrogate of the same degree.
    We support this claim by proving that the least squares fitting of the Walsh series with the original penalty function indeed yield the Walsh coefficients. 
    \item For the exponential penalty, we obtain its Walsh coefficients in a closed-form, i.e. its Walsh coefficients can be calculated directly from the constraint coefficients,
    without iterating through the long summation of WHT.
    \item We identify structural shortcuts that shorten the exact spectrum before any truncation is applied: a constraint acting on $r$ of the $n$ variables has at most $2^{r}$ nonzero coefficients, a cardinality-type constraint has only $n+1$ distinct coefficient values, and a step-like penalty at a balanced threshold has no even-degree coefficients at all.
\end{enumerate}

The rest of the paper is organized as follows. Sec.~\ref{sec:background} fixes the notation and reviews constrained binary optimization,
QAOA, the custom penalty formulation, and the Walsh--Hadamard transform.
Sec.~\ref{sec:method} develops the Walsh custom penalty QAOA (WCP-QAOA) for a general penalty function: the spectral phase separator,
the truncated Walsh spectrum, and an algorithm to compute Walsh coefficients in sub-exponential time.
Sec.~\ref{sec:results} shows the results to analyze truncated Walsh series and the WCP-QAOA tested on a 6-qubit Knapsack instance.
Sec.~\ref{sec:discussions} discusses some interesting properties, e.g. the $L^2$ optimality of the truncated Walsh spectrum, the closed-form spectrum for the exponential penalty,
and some structural shortcuts that can reduce the computational cost for WHT.
We conclude the paper in Sec.~\ref{sec:conclusion}.

\section{Background}
\label{sec:background}
We first disambiguate the readers by introducing some of the common conventions.
Characters with boldfaces are vectors, which also include collections of numbers, e.g. $\bx$ is a bit-string which is a collection of binary numbers.
The upright $\I$ denotes the imaginary unit $\I = \sqrt{-1}$, while the italic $i$ usually denotes an index or a counter for a sum or a product.
Calligraphic fonts such as $\mathcal{H}$ do not have a special meaning
and are often used as a replacement when the original symbols are already used in other venues to retain their nuance.

\subsection{Constrained binary optimization on quantum computers}

Consider a constrained binary program
\begin{align}
\min_{\bx \in \{0,1\}^n} \quad & f(\bx) \label{eq:obj}\\
\text{s.t.} \quad & g_i(\bx) = 0, \qquad i = 1, \dots, L, \label{eq:eq}\\
& h_j(\bx) \le 0, \qquad j = 1, \dots, K, \label{eq:ineq}
\end{align}
with objective $f$, equality constraints $g$ and inequality constraints $h$. A bit-string satisfying all constraints is \emph{feasible}; otherwise it is \emph{infeasible}.
We assume throughout that the constraints are linear,
\begin{equation}
h_j(\bx) = \ba_j^{\top}\bx - b_j = \sum_{i=1}^{n} a_{ji}x_i - b_j ,
\label{eq:linear-constraint}
\end{equation}
with nonnegative integer coefficients $a_{ji} \ge 0$, which is the case for the capacity constraints of knapsack-type problems.

The constrained program is usually converted to the QUBO loss function to be minimized with quantum algorithms:
\begin{equation}
\begin{split}
\mathcal{L}_{\mathrm{slack}}(\bx) = {}& f(\bx) + \sum_{i=1}^{L}\mu_i\, g_i(\bx)^2\\
&+ \sum_{j=1}^{K}\lambda_j\Big[h_j(\bx) + \sum_{k=0}^{n_y-1} 2^{k}y_{k,j}\Big]^2 ,
\end{split}
\label{eq:qubo}
\end{equation}
with tunable penalty factors $\mu_i,\lambda_j > 0$. The second term raises the loss whenever an equality is violated.
The third term does the same for the inequalities, but only after each inequality has been converted into an equality,
\begin{equation}
h(\bx) \le 0 \;\Longrightarrow\; h(\bx) + y = 0,
\label{eq:slack}
\end{equation}
by the introduction of a \emph{slack variable} $y\geq 0$, which must itself be written in binary as $\sum_{k=0}^{n_y-1}2^k y_k$ with $n_y = \lceil \log_2 y_{\max}\rceil$ 
and $y_{\max}$ the largest value $y$ can take.
Because $x_i^2 = x_i$ for binary variables, $\mathcal{L}_{\mathrm{slack}}$ contains only a constant, linear terms and quadratic terms, and is therefore a genuine QUBO.

The slack variables are the antagonists of this story. They add $n_y$ variables (or qubits) per inequality, and they degrade the landscape:
for a fixed feasible $\bx$ only one value of $s$ satisfies Eq.~\eqref{eq:slack}, so the great majority of slack assignments penalize a perfectly feasible $\bx$.
Our interest is in the \emph{slack-free} formulations that avoid them entirely.

Moving towards the quantum setting, the binary variables are mapped to spin variables $z_i \in \{+1,-1\}$ by the linear map
\begin{equation}
    x_i = \frac{1 - z_i}{2},
\label{eq:spinmap}
\end{equation}
and the spins are replaced by Pauli-$Z$ operators, $z_i \to Z_i$, with
\begin{align}
    Z & = \begin{pmatrix}
        1 & 0 \\
        0 & -1
    \end{pmatrix} \\
    Z_i & = I\otimes \cdots \otimes \underbrace{Z}_{i\text{-th position}}\otimes \cdots \otimes I,
\end{align}
where $\otimes$ is the Kronecker product. For a quadratic loss, this produces an Ising Hamiltonian which defines the problem:
\begin{equation}
    H_C = \sum_{i<j} q_{ij} Z_i Z_j + \sum_{i} l_i Z_i + \text{const},
    \label{eq:ising}
\end{equation}
which is diagonal in the computational basis. $q_{ij}$ and $l_i$ are coefficients. The eigenvalue attached to $\ket{\bx}$ equal to the loss of that bit-string,
\begin{equation}
    H_C\ket{\bx} = \mathcal{L}(\bx)\ket{\bx}.
    \label{eq:diag}
\end{equation}
The task of a variational algorithm is then to prepare a trial state $\ket{\psi(\bm{\theta})}$ minimizing $\braket{\psi(\bm{\theta})|H_C|\psi(\bm{\theta})}$.

\subsection{QAOA and VQE}

QAOA~\cite{farhi2014quantum} prepares the trial state
\begin{equation}
\ket{\psi(\bm{\gamma},\bm{\beta})} = e^{-\I\beta_p H_M}e^{-\I\gamma_p H_C}\cdots e^{-\I\beta_1 H_M}e^{-\I\gamma_1 H_C}\ket{+}^{\otimes n},
\label{eq:qaoa}
\end{equation}
where $H_C$ is the problem Hamiltonian, $H_M = \sum_{i=1}^{n} X_i$ is the mixer Hamiltonian, $p$ is the QAOA depth (number of QAOA layers),
and $(\bm{\gamma},\bm{\beta}) \in \mathbb{R}^{2p}$ are the variational parameters. The unitary $e^{-i\gamma H_C}$ is the \emph{phase separator}:
it is diagonal and imprints a phase proportional to the loss of each bit-string. Eq.~\eqref{eq:qaoa} is the Trotterized and discretized form of
adiabatic quantum computation~\cite{farhi:qaa}, which is the origin of both the theoretical guarantee at large $p$ and the empirical regularity of the optimal parameters
exploited by initialization heuristics~\cite{leo2020,tqa2021}.
The quantum alternating operator ansatz~\cite{new-qaoa} generalizes Eq.~\eqref{eq:qaoa} by allowing mixers other than the transverse field;
the $XY$ mixer~\cite{xy-mixer,constraint-preserving-qaoa} preserves the subspace of constant Hamming weight,
and the Grover mixer~\cite{grover-mixer} preserves the support of a chosen initial superposition, so that infeasible states are never populated.

VQE~\cite{vqe-first,vqe-review} minimizes the same expectation but places no restriction on the form of the circuit; hardware-efficient ansatz~\cite{hea} are common.
The crucial structural difference is that VQE never places $H_C$ in an exponent, whereas QAOA does.

\subsection{Custom penalty formulation}

Instead of squaring the inequalities and introducing slack, the custom penalty formulation~\cite{step-qce,cvar-cp} applies an indicator function 
$\xi:\mathbb{R}\to\mathbb{R}$ directly to the constraint value,
\begin{equation}
\mathcal{L}_{\mathrm{custom}}(\bx) = f(\bx) + \sum_{j=1}^{K}\lambda_j\, \xi\big[h_j(\bx)\big].
\label{eq:custom}
\end{equation}
A good $\xi$ is small for $h_j(\bx)\le 0$ and large for $h_j(\bx) > 0$. Two canonical choices are the Heaviside step function
\begin{equation}
    \Theta(y) = 
    \begin{cases} 
        1 & y > 0,\\
        0 & y \le 0,
    \end{cases}
\label{eq:heaviside}
\end{equation}
and the exponential penalty $e^{\eta y}$ with $\eta>0$, which is the one used implicitly by unbalanced penalization. The corresponding Hamiltonian is
\begin{equation}
H_C = H_f + \sum_{j=1}^{K}\lambda_j\, \xi(H_j),
\label{eq:custom-ham}
\end{equation}
where $H_f$ and $H_j$ are obtained from $f$ and $h_j$ by the map~\eqref{eq:spinmap}. Since $H_j$ is diagonal, $\xi$ acts entry-wise on its eigenvalues,
and no polynomial representation of $\xi$ is needed to evaluate Eq.~\eqref{eq:custom-ham}.

If the penalty factors are chosen large enough, Eq.~\eqref{eq:custom} separates feasible from infeasible solutions exactly,
in the sense that every feasible bit-string has strictly lower loss than every infeasible one, so the global minimum of $\mathcal{L}_{\mathrm{custom}}$ is the constrained optimum.
For the step penalty, it suffices that $\lambda$ exceed the spread of the objective, which can be bounded by the standard upper-bound rule~\cite{penalty-weights}.

The challenge in using Eq.~\eqref{eq:custom-ham} in QAOA is that $H_C$ now appears in the exponent of Eq.~\eqref{eq:qaoa}.
A good indicator function is usually nonlinear and, in the cases of interest, non-polynomial, and there is no elementary gate decomposition of $e^{-\I\gamma\,\xi(H_j)}$.
Arithmetic circuits with ancillary registers and comparators can implement the nonlinear indicator function~\cite{qc-arithmetic,addition-qc,ripple-carry-circuit},
but they reinstate the qubit overhead that the formulation was designed to avoid.

\subsection{The Walsh--Hadamard transform}
\label{sec:wht}

The Walsh--Hadamard transform uniquely expresses a function with Boolean inputs $\phi:\{0,1\}^n\to\mathbb{R}$ in the Walsh basis~\cite{odonnell-boolean-func}:
\begin{equation}
\phi(\bx) = \sum_{S\subseteq[n]} \hat{\phi}(S)\,\chi_S(\bx),
\label{eq:wht}
\end{equation}
where $[n]=\{1,\dots,n\}$ and the Walsh characters are the parity monomials
\begin{equation}
\chi_S(\bx) = \prod_{i\in S} z_i = (-1)^{\sum_{i\in S} x_i},
\label{eq:character}
\end{equation}
with $z_i = (-1)^{x_i}$ and by definition $\chi_\varnothing \equiv 1$.
The Walsh coefficients are computed as 
\begin{equation}
    \hat{\phi}(S) = \frac{1}{2^n}\sum_{\bx\in\{0,1\}^n}\phi(\bx)\chi_S(\bx)
    \label{eq:coeff}
\end{equation}
The \emph{degree} of $\phi$ is the largest $|S|$ with $\hat{\phi}(S)\neq 0$, and we call the terms with $|S|=k$ the \emph{$k$-local} terms.

Recall that with a linear mapping Eq.~\eqref{eq:spinmap}, $\chi_S(\bx)$ can be converted to $\chi_S(\bz)$ in the spin domain.
Interestingly, the parities of all the spin-strings are encoded in the Hadamard matrix
\begin{equation}
    \mathcal{H} =
    \begin{pmatrix} 
        1 & 1\\
        1 & -1
    \end{pmatrix}.
    \label{eq:hadamard}
\end{equation}
We use the notation $\mathcal{H}^n := \mathcal{H}^{\otimes n}$.
The rows (or columns) of $\mathcal{H}^n$ are (up to normalization) the characters $\chi_S$ evaluated on all $2^n$ bit-strings:
\begin{equation}
    \mathcal{H}^n =
    \begin{pmatrix}
        \line(1,0){30} & \bm\chi_\varnothing & \line(1,0){30} \\ 
        \line(1,0){30} & \bm\chi_{\{1\}} & \line(1,0){30} \\
        & \vdots & \\
        \line(1,0){30} & \bm\chi_{\{1,2\}} & \line(1,0){30} \\
        & \vdots & \\
        \line(1,0){30} & \bm\chi_{[n]} & \line(1,0){30} 
    \end{pmatrix}
\end{equation}
Therefore, from Eq.~\eqref{eq:wht}, the Walsh coefficient vector can be obtained by solving the system of linear equations
\begin{align}
    \mathcal{H}^n \hat{\bm{\phi}} & = \bm{\phi} \label{eq:walsh-lse} \\
    \hat{\bm{\phi}} & = \frac{1}{2^n}\mathcal{H}^n\bm{\phi} \label{eq:walsh-lse-inv},
\end{align}
since $(\mathcal{H}^n)^{-1} = 2^{-n}\mathcal{H}^n$. Note that the size of $\mathcal{H}^n$ is exponentially large, and solving Eq.~\eqref{eq:walsh-lse-inv} requires the knowledge
of $\phi(\bx)$ for all the $2^n$ bit-strings.
Alternatively, the coefficients can be estimated from Eq.~\eqref{eq:coeff}
\begin{equation}
    \hat{\phi}(S) = \Exp_{\bx \sim\{0,1\}^n}\big[\phi(\bx)\chi_S(\bx)\big].
    \label{eq:phi-est}
\end{equation}
This avoids enumerating all $2^n$ values of $\phi$ but requires a number of samples that grows with the accuracy demanded of small coefficients.

The Walsh basis is a natural choice because it is indeed the Pauli-$Z$ basis for diagonal operators~\cite{boolean-ham,welch-eff-qc}.
Under $z_i\to Z_i$ the character $\chi_S$ becomes the Pauli string $Z_S = \prod_{i\in S}Z_i$, and Eq.~\eqref{eq:wht} becomes an operator identity.

\begin{lemma}[Spectral form of a diagonal Hamiltonian]
\label{lem:dictionary}
Let $H_\phi$ be the diagonal operator defined by $H_\phi\ket{\bx} = \phi(\bx)\ket{\bx}$. Then
\begin{equation}
H_\phi = \sum_{S\subseteq[n]}\hat{\phi}(S)\, Z_S ,
\label{eq:ham-wht}
\end{equation}
and all terms commute, so for every $\gamma$
\begin{equation}
e^{-\I\gamma H_\phi} = \prod_{S\subseteq[n]} e^{-\I\gamma\hat{\phi}(S)Z_S}
\label{eq:exact-factorization}
\end{equation}
exactly, with no Trotter error and in any ordering. $Z_S = \prod_{i\in S}Z_i$ is the Pauli string acting on the qubits with indices in $S$.
\end{lemma}

\begin{proof}
Both sides of Eq.~\eqref{eq:ham-wht} are diagonal, and their $\bx$-th entries agree by Eq.~\eqref{eq:wht} together with $Z_S\ket{\bx} = \chi_S(\bx)\ket{\bx}$. All $Z_S$ are diagonal and therefore mutually commuting, which gives Eq.~\eqref{eq:exact-factorization}.
\end{proof}

Lemma~\ref{lem:dictionary} is the whole reason the program of this paper is possible: it converts the problem of implementing a non-polynomial penalty in the QAOA exponent
into the problem of finding a short Walsh series for it.
We call this framework the Walsh custom penalty QAOA (WCP-QAOA). 

\section{Custom Penalty QAOA with Walsh--Hadamard Spectrum}
\label{sec:method}
Throughout this section we consider a single linear inequality constraint,
\begin{equation}
h(\bx) = \ba^{\top}\bx - b \le 0,
\label{eq:one-constraint}
\end{equation}
with $\ba \in \mathbb{Z}_{\geq 0}^n$ and $0 \le b < \onenorm{\ba}$. The associated penalty is
\begin{equation}
\phi(\bx) = \xi\big(\ba^{\top}\bx - b\big).
\label{eq:phi}
\end{equation}
The case of multiple constraints is treated in Sec.~\ref{sec:multi}. We write $s = \ba^{\top}\bx \in \{0,1,\dots,\onenorm{\ba}\}$ for the constraint value,
so that $\phi$ depends on $\bx$ only through $s$, and $\onenorm{\ba} = \sum_i a_i$ is the $L^1$ norm of $\ba$ (the largest value $s$ can attain).

\subsection{Spectral phase separator for QAOA}

Applying Lemma~\ref{lem:dictionary} to Eq.~\eqref{eq:phi} gives the exact Walsh-expanded penalty Hamiltonian
\begin{equation}
    H_\phi = \sum_{S\subseteq[n]}\hat{\phi}(S)Z_S,
\label{eq:Hphi}
\end{equation}
where
\begin{equation}
    \hat{\phi}(S) = \frac{1}{2^n} \sum_{\bx\in\{0,1\}^n} \xi(\ba^{\top}\bx - b)\chi_S(\bx).
\end{equation}
The complete problem Hamiltonian for Eq.~\eqref{eq:custom} is $H_C = H_f + \lambda H_\phi$. The phase separator factorizes exactly as in Eq.~\eqref{eq:exact-factorization},
and each factor (corresponding to a term in the Walsh series) is a standard $\mathrm{CNOT}$-ladder-and-$R_z$ block,
\begin{multline}
    e^{-\I\gamma\hat\phi(S)Z_S}
    = \left(\prod_{\text{ladder}} \mathrm{CNOT}\right) \\
    \cdot R_{z}^{(m)} \big(2\gamma\hat\phi(S)\big)
    \left(\prod_{\text{ladder}} \mathrm{CNOT}\right),
    \label{eq:zsblock}
\end{multline}
where $R_z^{(m)}(\theta) = e^{-\I\theta Z_m/2}$ is the $Z$-rotation acting on the pivot qubit $m = \max(S)$ with angle $\theta$. 
This circuit requires $2(|S|-1)$ $\mathrm{CNOT}$s and one single-qubit rotation. There are no ancillas, no comparators and no Trotter error.
Fig.~\ref{fig:z-ladder} shows the implementation of such an operator with the example of $S=\{1,2,4\}$. 

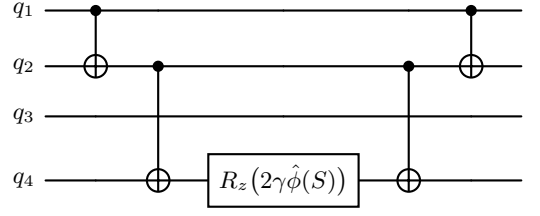
\begin{figure}
    \centering
    \begin{quantikz}
    	\lstick{$q_1$} & \ctrl{1} & 		 &                                           &          & \ctrl{1} & \\
    	\lstick{$q_2$} & \targ{}  & \ctrl{2} &                                           & \ctrl{2} & \targ{}  & \\
    	\lstick{$q_3$} &          &          &                                           &          &          & \\
    	\lstick{$q_4$} &          & \targ{}	 & \gate{R_z \big(2\gamma\hat{\phi}(S)\big)} & \targ{}  &          &
    \end{quantikz}
    \caption{Quantum circuit to implement the QAOA phase separator $e^{-\I\gamma\hat{\phi}(S)Z_S}$, with $S = \{1,2,4\}$. }
    \label{fig:z-ladder}
\end{figure}

Now we look at the circuit to implement the entire Walsh series in Eq.~\eqref{eq:exact-factorization}.
Since all Pauli-$Z$ strings commute, the phase separator for the Walsh Hamiltonian $H_\phi$ factorizes exactly without Trotter error.
The identity term $S=\varnothing$ contributes only a global phase and can be dropped.
Since the $R_z$ rotation gate is often implemented as a virtual gate (a phase shift pulse) in quantum hardware~\cite{z-gate}, we only consider the number of CNOTs required.
Concatenating all $2^n-1$ non-identity terms in arbitrary order requires $\sum_{k=2}^{n}2(k-1)\binom{n}{k} = n2^n-2^{n+1}+2$ CNOTs, i.e., $\BigO(n2^n)$ gates in total.

However, Bullock and Markov pointed out that most of these CNOTs are redundant~\cite{bullock-markov}.
For two consecutive terms $(S,S')$ sharing a pivot, the intervening CNOT ladders cancel except for CNOTs controlled on the
symmetric difference $S\triangle S'$, the bitwise XOR of their binary labels.
Welch \emph{et al.} exploited this by ordering the non-identity subsets along a binary reflected Gray code (known as the sequency order in~\cite{welch-eff-qc}),
so consecutive subsets differ in exactly one element~\cite{welch-eff-qc}.
Then, the sequence is partitioned by pivot, $G_i=\{S:\max(S)=i\}$, so that each block places its rotations on, and
targets its CNOTs at, the single qubit $i$. For $n=3$ the Gray-code sequence
$\{1\},\{1,2\},\{2\},\{2,3\},\{1,2,3\},\{1,3\},\{3\}$ splits into
$G_1=\{\{1\}\}$, $G_2=\{\{1,2\},\{2\}\}$ and
$G_3=\{\{2,3\},\{1,2,3\},\{1,3\},\{3\}\}$. Every block ends on the single-qubit
term $\{i\}$, so the parity accumulated on qubit $i$ is undone within the block
and no trailing CNOTs are needed. Block $G_i$ contains $2^{i-1}$ terms, each
preceded by exactly one CNOT for $i\ge2$, while $G_1$ is a bare rotation.
Summing over blocks, we get the number of CNOTs as $\sum_{i=2}^{n}2^{i-1}=2^n-2$,
which reproduces the gate count in~\cite{bullock-markov}, drastically reduced from $\BigO(n2^n)$.
At $n=12$ this is a reduction from $40962$ to $4094$ CNOTs.
This is still an enormous number of gate counts, and it scales exponentially with the number of qubits.
Therefore, it calls for a truncation method, where we can control the number of terms in the Walsh series and hence the number of gates under a polynomial scaling.


\subsection{Truncated Walsh series}
The natural approximation is to keep only the low-degree part of the spectrum. We define the degree-$k$ truncation
\begin{equation}
    \phi_k(\bx) = \sum_{|S|\le k}\hat{\phi}(S)\chi_S(\bx)
    \label{eq:trunc-phi}
\end{equation}
and the corresponding truncated Walsh Hamiltonian
\begin{equation}
    H_\phi^{(k)} = \sum_{|S|\le k}\hat{\phi}(S)Z_S.
    \label{eq:trunc-ham}
\end{equation}
We use $\bphi$ to denote the vector form of $\phi$, and $T_k\bphi$ to denote the vector form of $\phi_k$. 
Recall that the Walsh characters $\chi_S$ are the bases for $\phi$. 
As $\chi_S$ are orthonormal, $T_k\bphi$ is precisely the orthogonal projection of $\phi$ onto the space of polynomials of degree at most $k$, i.e.,
it is the \emph{best} degree-$k$ approximation of $\bphi$ in $L^2$ under the uniform measure. We show this in Sec.~\ref{sec:walsh-opt}.
Two remarks are in order. First, the truncation is applied to the Hamiltonian, not to the unitary: Eq.~\eqref{eq:exact-factorization} still holds exactly for $H_\phi^{(k)}$,
so the only error incurred is the classical function-approximation error $\bphi - T_k\bphi$. Second, $H^{(k)}_\phi$ contains
\begin{equation}
    M_k = \sum_{i=1}^{k}\binom{n}{i} = \BigO(n^k)
    \label{eq:Mk}
\end{equation}
Pauli terms, which is polynomial in $n$ for fixed $k$. In particular $k=2$ gives $n(n+1)/2$ terms, exactly the size of a dense QUBO phase separator.

In terms of the circuit cost, no closed-form optimal CNOT count is available for a truncated Walsh series.
Welch \emph{et al.} suggest that CNOT commutation identities reduce the Gray-code construction to $O(N')$ gates for $N'$ nonzero coefficients in the Walsh series,
but support this only numerically for up to five qubits~\cite{welch-eff-qc}.
Zylberman \emph{et al.} note that Gray-code ordering need not be efficient for sparse sets~\cite{zylberman-cnot}.
More fundamentally, synthesizing a CNOT-minimal circuit realizing a prescribed set of parities (a \emph{parity network}) is NP-complete already when all CNOTs are
restricted to a common target~\cite{amy-cnot}, which is the evidence that no general closed form should be expected. We therefore bound the cost instead.
Each CNOT changes the parity carried by a single wire and can thus create at most one new parity, while for generic angles every $Z_S$ with $\hat\phi(S)\neq0$ must be
carried by some wire when its rotation is applied~\cite{amy-cnot}.
Therefore, for a $k$-local truncation, any ancilla-free implementation satisfies
\begin{equation}
    \sum_{i=2}^{k}\binom{n}{i}\le N_{\mathrm{CNOT}}\le
    \sum_{i=2}^{k}2(i-1)\binom{n}{i},
    \label{eq:cnot-sandwich}
\end{equation}
the upper bound being the term-by-term construction where no CNOT cancellation occur at all.
For $k=2$ this pins the cost within a factor of two, and in practice the GraySynth heuristic of
Ref.~\cite{amy-cnot} attains $\binom{n}{2}+n-1$ CNOTs, including the uncomputation required before the mixer.

How much is lost by truncating is a question about spectral concentration, and for a general penalty it is answered by two elementary identities:
\begin{align}
    \supnorm{\bphi - T_k\bphi} & \le \sum_{|S|>k}\big|\hat{\phi}(S)\big| ,
    \label{eq:err-linf}\\
    \twonorm{\bphi - T_k\bphi}^{2} & = \sum_{|S|>k}\hat{\phi}(S)^{2} ,
    \label{eq:err-l2}
\end{align}
the first by the triangle inequality with $|\chi_S| = 1$; the second by Parseval's identity
\begin{equation}
    \twonorm{\bphi}^2 = \sum_{S\subseteq[n]} \hat{\phi}(S)^2,
    \label{eq:parseval}
\end{equation}
which states that the squared 2-norm for any function $\phi$ is equal to the sum of its squared coefficients.
The design question is therefore how quickly the Walsh mass of a given penalty decays with degree. We define
\begin{equation}
    \Pi_k = \frac{\sum_{|S|\le k}\hat{\phi}(S)^{2}}{\sum_{S\subseteq[n]}\hat{\phi}(S)^{2}} \in [0,1]
    \label{eq:retained-power}
\end{equation}
as the fraction of spectral power retained at degree $k$.
This is a convenient figure of merit for comparing different candidate penalty profiles at a fixed term budget.
It is a property of the pair (penalty, constraint) rather than of the penalty alone, since the coefficients $\ba$ determine which bit-strings share a constraint value,
and hence how the penalty values are distributed over the Boolean hypercube.

Truncation perturbs the penalty, and the first thing to ask of the perturbed Hamiltonian is whether its ground state still solves the constrained problem.
Let the spread of the objective be
\begin{equation}
    \Delta_f = \max_{\bx}f(\bx) - \min_{\bx}f(\bx).
\end{equation}
Let the \emph{penalty gap} be 
\begin{equation}
    \delta_\xi = \min_{\bx:\,h(\bx)>0}\phi(\bx) - \max_{\bx:\,h(\bx)\le 0}\phi(\bx),
\label{eq:gap}
\end{equation}
which is the gap between the least-penalized infeasible and the most-penalized feasible bit-string. Let
\begin{equation}
\varepsilon_k = \sum_{|S|>k} \big|\hat{\phi}(S) \big|
\label{eq:eps}
\end{equation}
be the sum of the discarded Walsh coefficients, the right-hand side of Eq.~\eqref{eq:err-linf}.

\begin{proposition}[Feasibility separation under truncation]
\label{prop:separation}
If
\begin{equation}
\lambda\big(\delta_\xi - 2\varepsilon_k\big) \;>\; \Delta_f ,
\label{eq:sep-general}
\end{equation}
then the truncated loss $\mathcal{L}_k(\bx) = f(\bx) + \lambda\,T_k\phi(\bx)$ satisfies $\mathcal{L}_k(\bx) < \mathcal{L}_k(\bx')$ for every feasible $\bx$ and every infeasible $\bx'$.
\end{proposition}

\begin{proof}
By Eq.~\eqref{eq:err-linf}, $T_k\phi(\bx)\le\phi(\bx)+\varepsilon_k$ and $T_k\phi(\bx')\ge\phi(\bx')-\varepsilon_k$. Hence $\mathcal{L}_k(\bx') - \mathcal{L}_k(\bx) \ge -\Delta_f + \lambda\big[\phi(\bx') - \phi(\bx) - 2\varepsilon_k\big] \ge -\Delta_f + \lambda(\delta_\xi - 2\varepsilon_k) > 0$.
\end{proof}

For the step penalty $\delta_\Theta = 1$ for every constraint, so Eq.~\eqref{eq:sep-general} asks only that the discarded weight be below one half and that
$\lambda$ exceed $\Delta_f/(1-2\varepsilon_k)$. Two caveats apply.
First, $\varepsilon_k$ is an $\ell_1$ tail and cannot in general be evaluated without the high-degree coefficients it sums.
Second, Eq.~\eqref{eq:sep-general} is sufficient and pessimistic,
because it demands that the truncated penalty order every feasible bit-string below every infeasible one.
In a variational setting this is more than is needed: the sampled distribution concentrates on low-energy states,
and under CVaR training~\cite{cvar,cvar-cp} only the lowest tail of the distribution enters the objective at all.
A tail-based criterion, whether the best feasible solution still appears among the lowest-energy states of the truncated Hamiltonian, is the operationally relevant one,
and we expect it to be met at considerably smaller $k$ than Eq.~\eqref{eq:sep-general} requires.


\subsection{Computing the coefficients in sub-exponential time}
\label{sec:coefficients}

Computing the Walsh coefficients is challenging, as Eq.~\eqref{eq:coeff} is unusable as written, since it sums over $2^n$ bit-strings. 
The sampling estimator of Eq.~\eqref{eq:phi-est} needs many samples to resolve small coefficients.
Solving the linear system of equations in Eq.~\eqref{eq:walsh-lse} requires the knowledge of the vector $\bphi$, which $\phi(\bx)$ for all $x$'s.
Here, we introduce a method to compute a single Walsh coefficient in sub-exponential time.

\begin{theorem}[Computation of the Walsh coefficient in sub-exponential time]
\label{thm:sub-exp-coeff}
The Walsh coefficient can be computed with
\begin{equation}
    \hat{\phi}(S) = \sum_{s=0}^{\onenorm{\ba}}\xi(s-b) \omega_S(s),
    \label{eq:signedfreq}
\end{equation}
where $\omega_S$ is the signed frequency
\begin{equation}
    \omega_S(s) = \frac{1}{2^n}\sum_{\bx:\ba^{\top}\bx = s}\chi_S(\bx)
    \label{eq:signedfreq-def}
\end{equation}
can be computed by 
\begin{equation}
    \bm\omega_S = \bm{g}_1 * \bm{g}_2 * \cdots * \bm{g}_n ,
    \label{eq:convolution}
\end{equation}
where $*$ denotes discrete convolution, and $\bm{g}_i$ is a vector with
\begin{equation}
    \bm{g}_i[s] =
    \begin{cases}
        \phantom{-}\tfrac{1}{2} & s = 0\\[2pt]
        -\tfrac{1}{2} & s = a_i \ \text{and}\ i\in S\\[2pt]
        \phantom{-}\tfrac{1}{2} & s = a_i \ \text{and}\ i\notin S\\[2pt]
        \phantom{-}0 & \text{otherwise}.
    \end{cases}
    \label{eq:gvec}
\end{equation}
Consequently $\hat\phi(S)$ can be computed in $\BigO(n\onenorm{\ba})$ time and $\BigO(\onenorm{\ba})$ memory.
\end{theorem}

Note that this computation is exact, not approximation, achieved by grouping and rearranging the sum.
Because $\phi$ depends on $\bx$ only through $s = \ba^{\top}\bx$, we can group the sum in Eq.~\eqref{eq:coeff} by the value of $s$.
We call $\omega_S(s)$ the \emph{signed frequency}: it counts how many bit-strings realize the constraint value $s$, each weighted by its parity on $S$
($+1$ for an even number of ones inside $S$, $-1$ for an odd number), and normalized by $2^n$. 
The vector $\omega_S \in \mathbb{R}^{\onenorm{\ba}+1}$ is a property of the constraint alone; the penalty enters only through the dot product in Eq.~\eqref{eq:signedfreq}.
This separation is useful in its own right, since a single computation of $\omega_S$ serves every candidate penalty function.

The proof is given in Appendix~\ref{app:convolution}. The complexity claim follows because each $\bm{g}_i$ has at most two nonzero entries, 
so convolving it with a vector of length $\onenorm{\ba}+1$ costs $\BigO(\onenorm{\ba})$, and there are $n$ such convolutions.
The final dot product in Eq.~\eqref{eq:signedfreq} costs $\BigO(\onenorm{\ba})$.

Combining Theorem~\ref{thm:sub-exp-coeff} with Eq.~\eqref{eq:Mk}, the classical preprocessing for the degree-$k$ Hamiltonian costs
\begin{equation}
    \BigO\big(n\onenorm{\ba}M_k\big) = \BigO\big(n^{k+1}\onenorm{\ba}\big),
    \label{eq:preprocessing}
\end{equation}
which is polynomial in $n$ for fixed $k$ and pseudo-polynomial in the constraint data.

\begin{remark}[The retained power is computable]
\label{rem:computable}
By Parseval, the denominator of Eq.~\eqref{eq:retained-power} is $\sum_S\hat\phi(S)^2 = \Ex_{\bx}[\phi(\bx)^2] = \sum_{s}\xi(s-b)^2\,\omega_\varnothing(s)$,
and $\omega_\varnothing$ is the case $S=\varnothing$ of Theorem~\ref{thm:sub-exp-coeff}.
The retained power $\Pi_k$ and the $L^2$ truncation error of Eq.~\eqref{eq:err-l2} can therefore be computed exactly at the cost of Eq.~\eqref{eq:preprocessing},
without evaluating a single coefficient of degree above $k$. The $\ell_1$ tail $\varepsilon_k$ of Eq.~\eqref{eq:eps} enjoys no such identity.
The Cauchy--Schwarz inequality gives the computable bound $\varepsilon_k \le \big(\sum_{j>k}\binom{n}{j}\big)^{1/2}\lVert\phi - T_k\phi\rVert_2$, which is useful only for small $n$.
\end{remark}

\subsection{Multiple constraints}
\label{sec:multi}

With $K$ inequality constraints the loss is Eq.~\eqref{eq:custom}, and each penalty is expanded independently,
\begin{equation}
H_C^{(k)} = H_f + \sum_{j=1}^{K}\lambda_j \sum_{|S|\le k}\hat{\phi}_j(S)\,Z_S ,
\label{eq:multi}
\end{equation}
where $\hat\phi_j$ is the spectrum of $\xi(\ba_j^{\top}\bx - b_j)$. Coefficients on the same $S$ from different constraints are summed, so the gate count does not grow with $K$; only the classical preprocessing does, linearly. The truncation error bounds extend by the triangle inequality: the supremum-norm error of the full penalty term is at most $\sum_{j}\lambda_j\varepsilon_k^{(j)}$, where $\varepsilon_k^{(j)}$ is the sum of the discarded coefficients, Eq.~\eqref{eq:eps}, of constraint $j$.

\subsection{Summary of the algorithm}

Collecting the preceding subsections, WCP-QAOA with a truncated Walsh spectrum proceeds as follows.

\begin{enumerate}
\item \textbf{Scale.} Rescale each constraint so that $\ba_j,b_j$ are integer with $\onenorm{\ba_j}$ as small as possible, since the cost of Step~3 is linear in $\onenorm{\ba_j}$.
\item \textbf{Choose the penalty function and the penalty factor.} Select a $\xi$ with a positive penalty gap $\delta_\xi$, Eq.~\eqref{eq:gap}, and a $\lambda_j$ large enough for 
Proposition~\ref{prop:separation} to be met, using a bound on the objective spread $\Delta_f$ such as the upper-bound rule~\cite{penalty-weights}.
\item \textbf{Compute the spectrum.} Check first whether the structural shortcuts of Sec.~\ref{sec:structure} apply, since they remove whole families of coefficients before any is computed.
Then, for increasing $k$, compute $\{\hat{\phi}_j(S)\}_{|S| = k}$ by Theorem~\ref{thm:sub-exp-coeff} at the cost of Eq.~\eqref{eq:preprocessing},
and stop when the retained spectral power $\Pi_k$ of Eq.~\eqref{eq:retained-power}, which is exactly computable (Remark~\ref{rem:computable}),
reaches the target or the gate budget of Eq.~\eqref{eq:cnot-sandwich} is exhausted. Individual terms whose coefficients fall below a magnitude threshold may additionally be discarded.
\item \textbf{Assemble the Hamiltonian.} Form $H_C^{(k)}$ as in Eq.~\eqref{eq:multi}, summing the coefficients that share the same $S$ across constraints.
\item \textbf{Run QAOA.} Build the phase separator from Eq.~\eqref{eq:zsblock} and optimize $(\bm{\gamma},\bm{\beta})$.
\end{enumerate}

Steps 1--4 are entirely classical and are performed once, before the circuit is executed.

\section{Numerical Results}
\label{sec:results}

We test our framework with a 6-qubit Knapsack problem. The 0-1 Knapsack problem is formulated as
\begin{align}
    \max_{\bx\in\{0,1\}^n}\ & \bv^\top\bx \label{eq:kp-obj} \\
    \text{s.t.}\ & \ba^\top \bx \leq b \label{eq:constraint}.
\end{align}
$\bx$ are the decision variables; $\bv$ contains the values (profit) of the items; $\ba$ contains the weights of the items; $b$ is the capacity of the Knapsack.
The problem has an integer inequality constraint, which introduces slack variables when converted to QUBO, and can be avoided using the CP formulation.
The 6-qubit (6-variable) instance has the following values:
\begin{align}
    \bm{v} & = [5, 10, 20, 30, 35, 40]^\top \\
    \bm{a} & = [3, 5, 6, 7, 9, 10]^\top \\
    b & = 14.
\end{align}
First, we investigate the effect of the truncated Walsh spectrum in approximating the original penalty function.
For all possible $2^6=64$ bit-strings, we plot the constraint values $\ba^\top\bx-b$ against its values in the penalty functions,
for the original function $\phi(\bx) = \xi(\ba^\top\bx -b)$ and the truncated Walsh spectrum $T_k\phi$.
Fig.~\ref{fig:walsh-approx} shows the results for different penalty functions: Heaviside step, hyperbolic tangent, exponential, and rectified linear unit (ReLU):
\begin{equation}
    \mathrm{ReLU}(y) = \max(0,y).
    \label{eq:relu}
\end{equation}
From the figure, the most intuitive observation is that the truncation fits better as degree $k$ increases.
The truncation degree $k$ represents the truncated Walsh series $T_k\phi$ has terms up to $k$-local, as defined in \eqref{eq:trunc-phi}, with the largest degree as $k=6$
representing the complete Walsh series.
On the other hand, for different penalty functions, same degree of truncation produces different level of approximation. 
For exponential and ReLU, $k=4$ is sufficient to approximate the full series well, whereas the Heaviside step function and the hyperbolic tangent appears to diverge more from it.
At $k=6$, all the values match the original penalty function as it is a full Walsh series.
Fig.~\ref{fig:pi-k} quantifies the level of divergence from the original penalty function. 
We plot $\tilde\Pi_k$ (stated in Eq.~\eqref{eq:retained-power}, but excluding the constant term $\hat\phi(\varnothing)$) 
\begin{equation}
    \tilde\Pi_k = \frac{\sum_{|S|\le k}\hat{\phi}(S)^{2} - \hat\phi(\varnothing)^2}{\sum_{S\subseteq[n]}\hat{\phi}(S)^{2} - \hat\phi(\varnothing)^2}
    \label{eq:retained-power-no-constant}
\end{equation}
against $k$. The constant term (offset) is usually large compared to other coefficients in the Walsh series,
and is removed to show clearer difference in the values for different functions.
For the exponential and ReLU penalties, $\tilde\Pi_k$ exceeds 0.99 at $k=3$, whereas for the step and the hyperbolic tangent penalties, $\Pi_k$ did not achieve 1.0 until $k=6$.
This shows that some functions require less Walsh terms to approximate, while others require more. 

\begin{figure*}
    \centering
    \subfloat[]{%
        \includegraphics[width=0.32\textwidth]{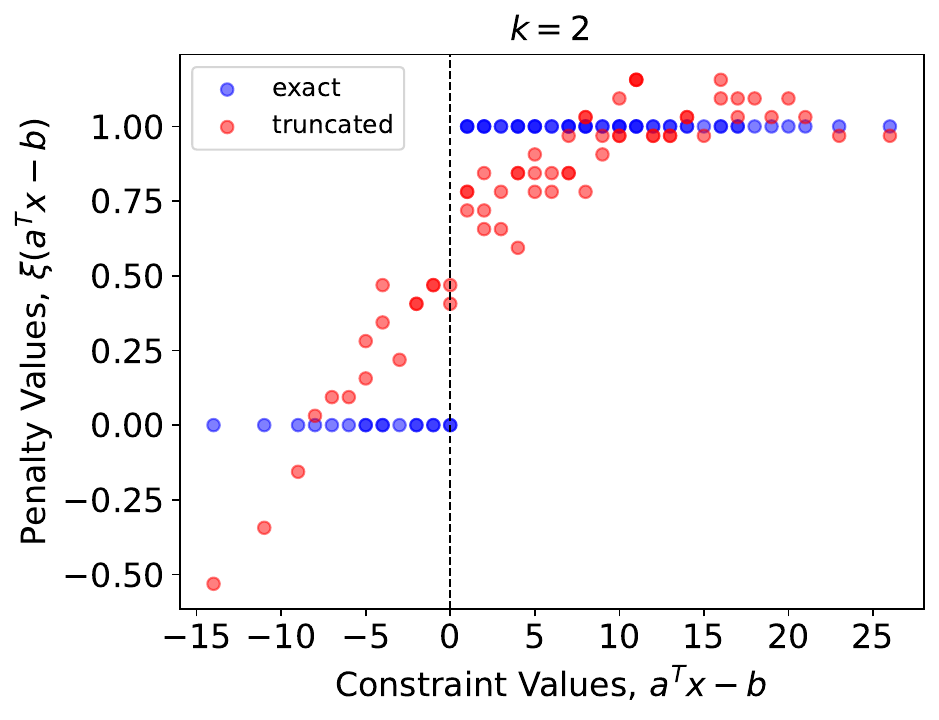}%
    }%
    \hfill 
    \subfloat[]{%
        \includegraphics[width=0.32\textwidth]{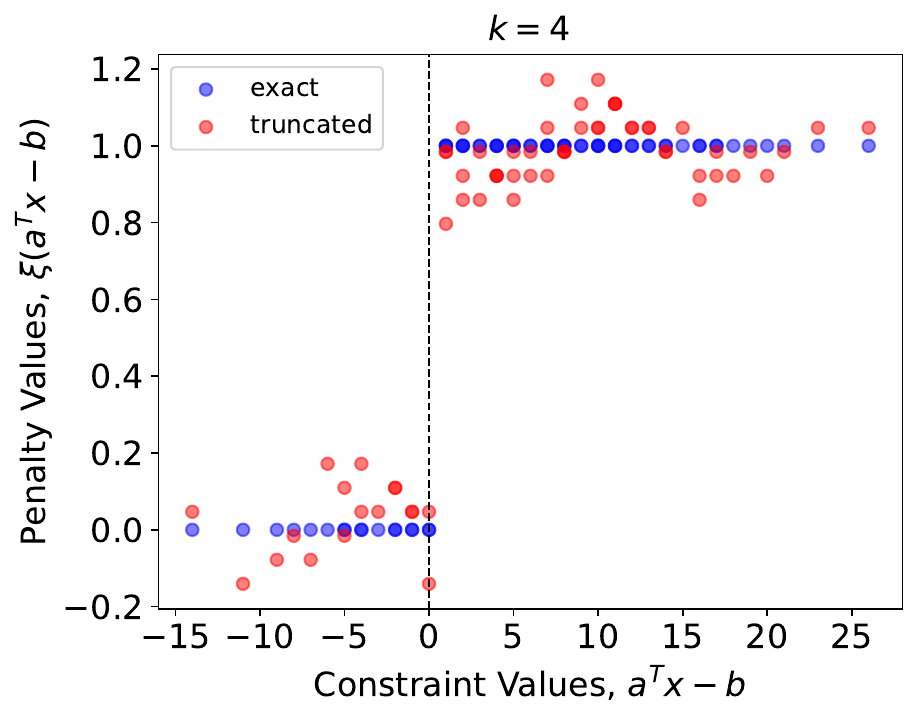}%
    }%
    \hfill
    \subfloat[]{%
        \includegraphics[width=0.32\textwidth]{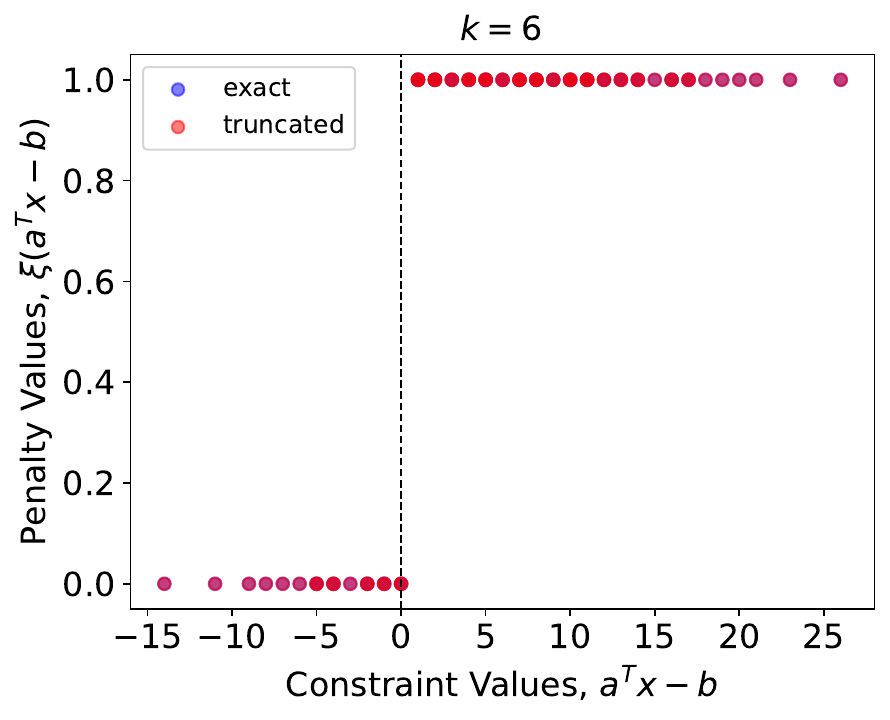}%
    }%
    \hfill
    \subfloat[]{%
        \includegraphics[width=0.32\textwidth]{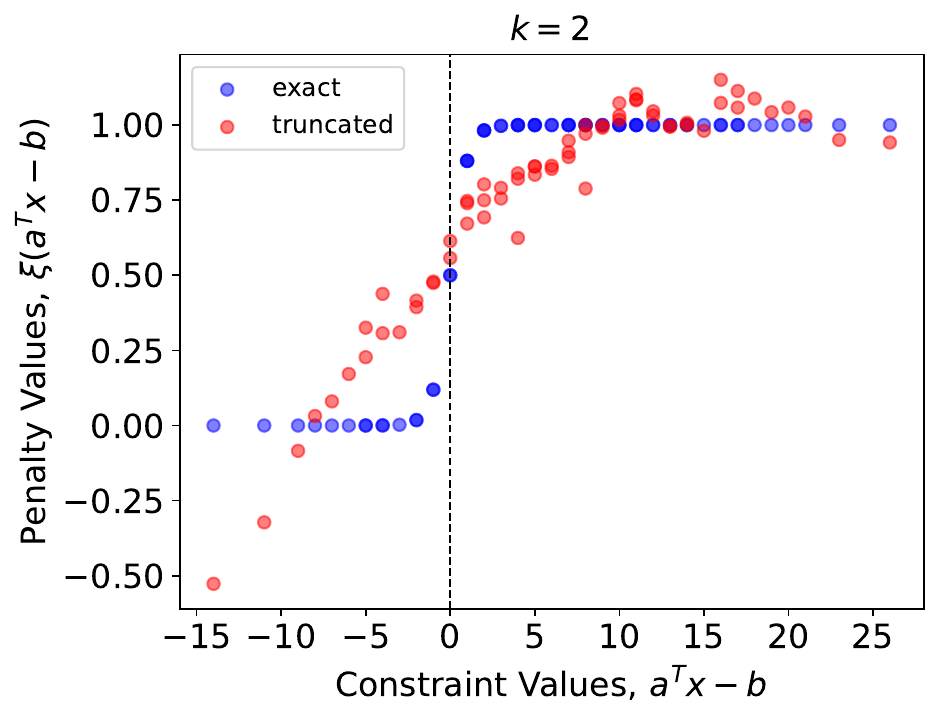}%
    }%
    \hfill
    \subfloat[]{%
        \includegraphics[width=0.32\textwidth]{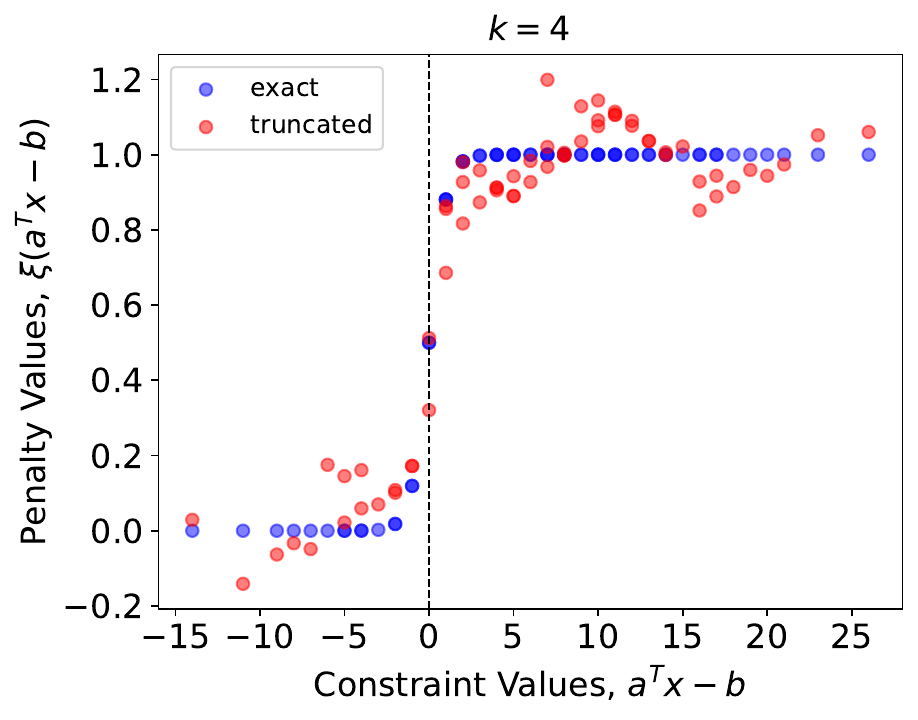}%
    }%
    \hfill
    \subfloat[]{%
        \includegraphics[width=0.32\textwidth]{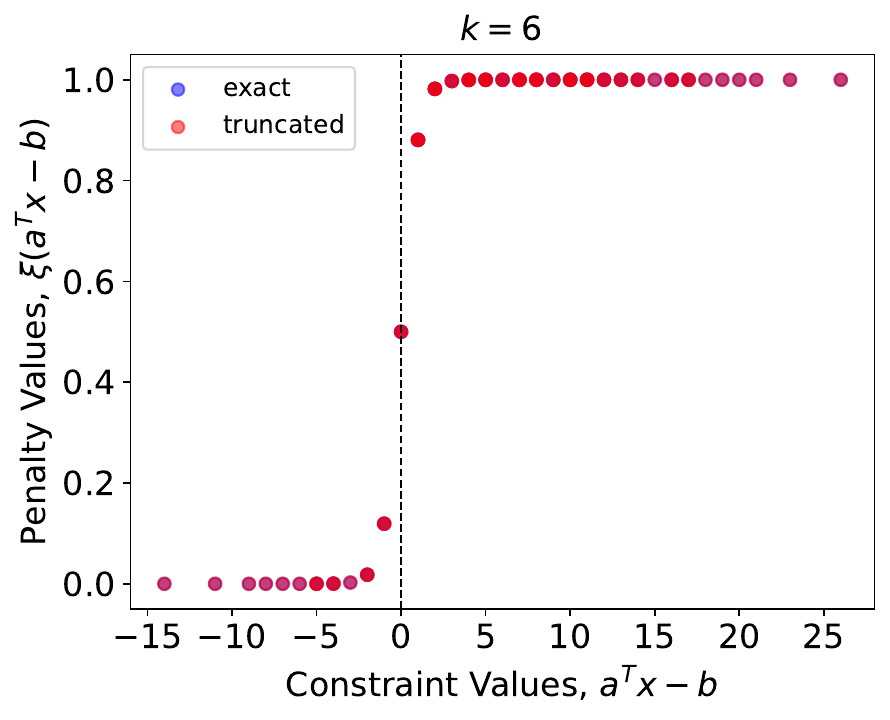}%
    }%
    \hfill
    \subfloat[]{%
        \includegraphics[width=0.32\textwidth]{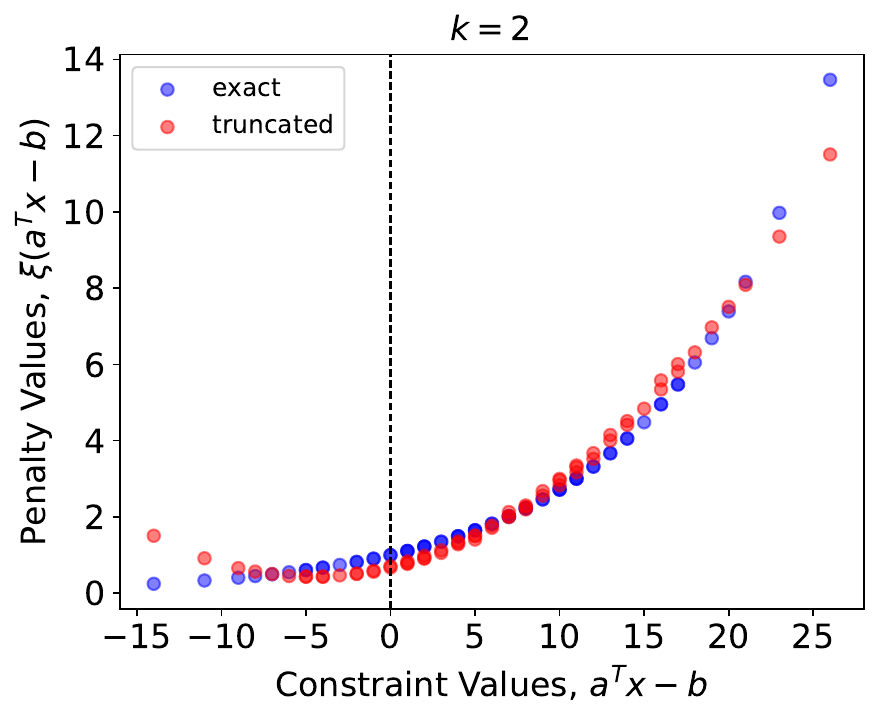}%
    }%
    \hfill
    \subfloat[]{%
        \includegraphics[width=0.32\textwidth]{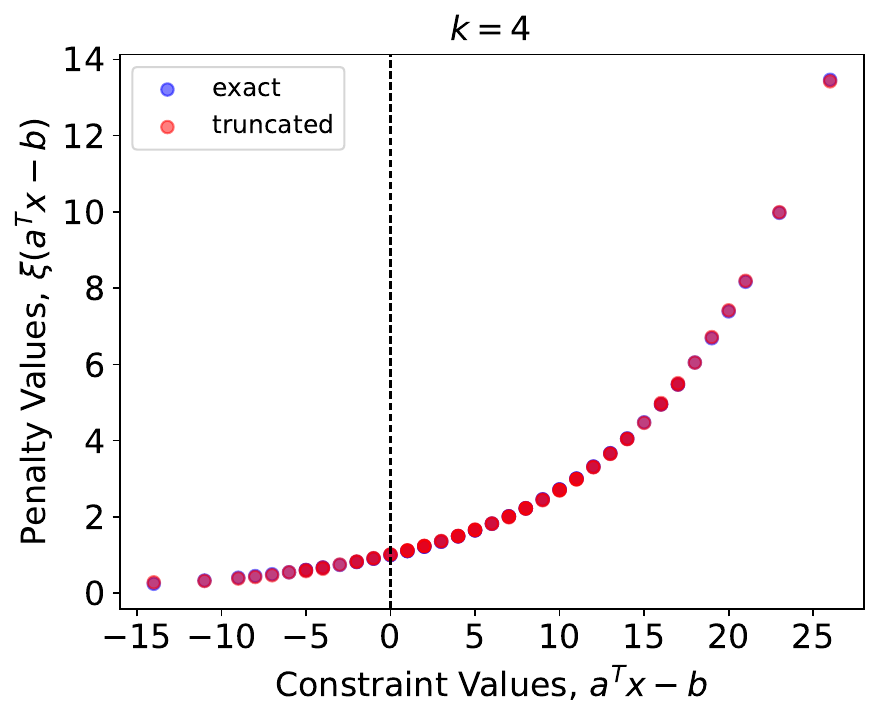}%
    }%
    \hfill
    \subfloat[]{%
        \includegraphics[width=0.32\textwidth]{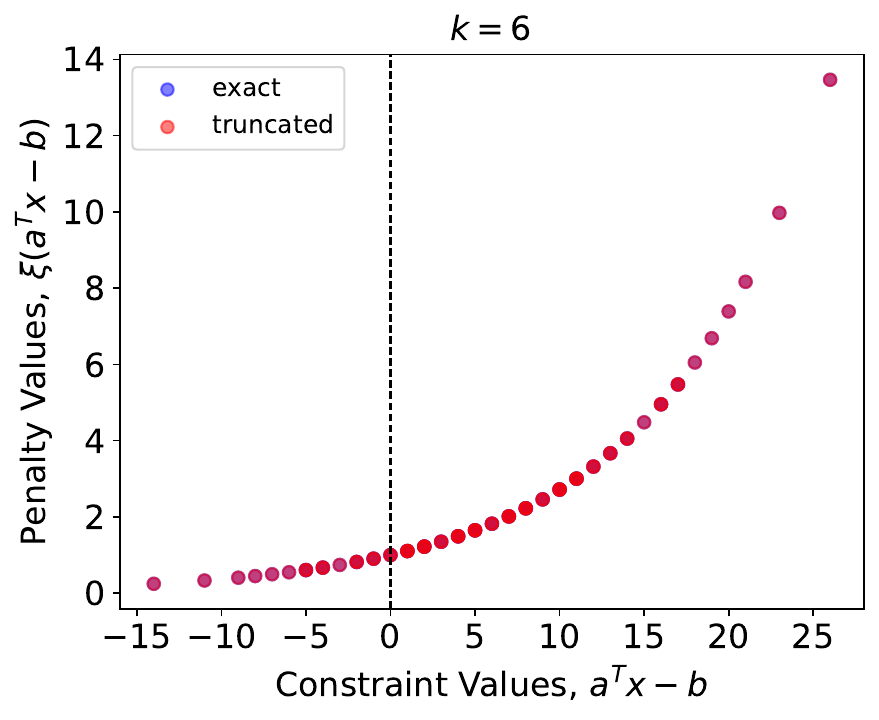}%
    }%
    \hfill
    \subfloat[]{%
        \includegraphics[width=0.32\textwidth]{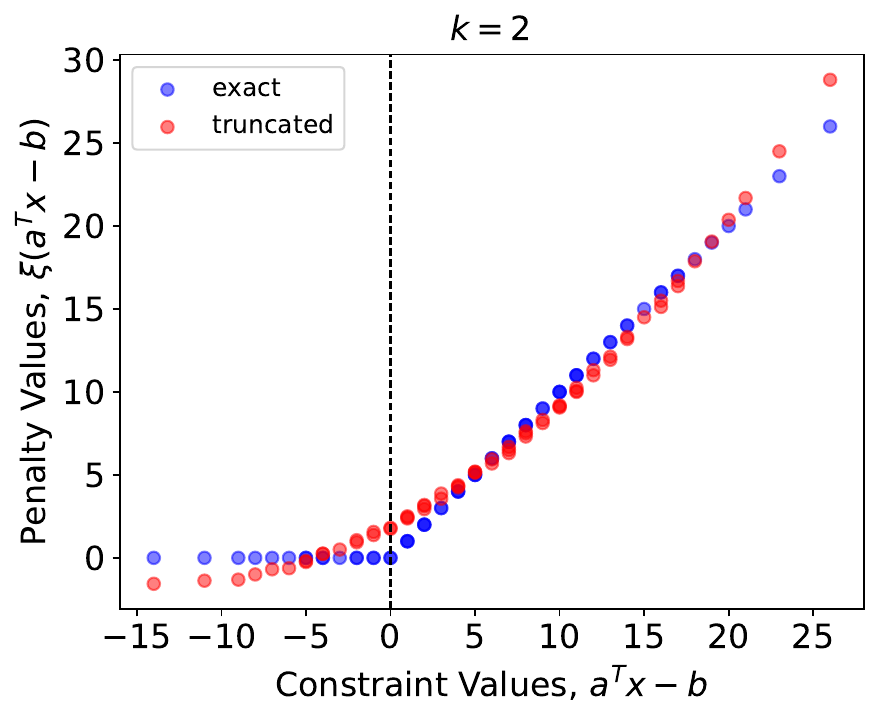}%
    }%
    \hfill
    \subfloat[]{%
        \includegraphics[width=0.32\textwidth]{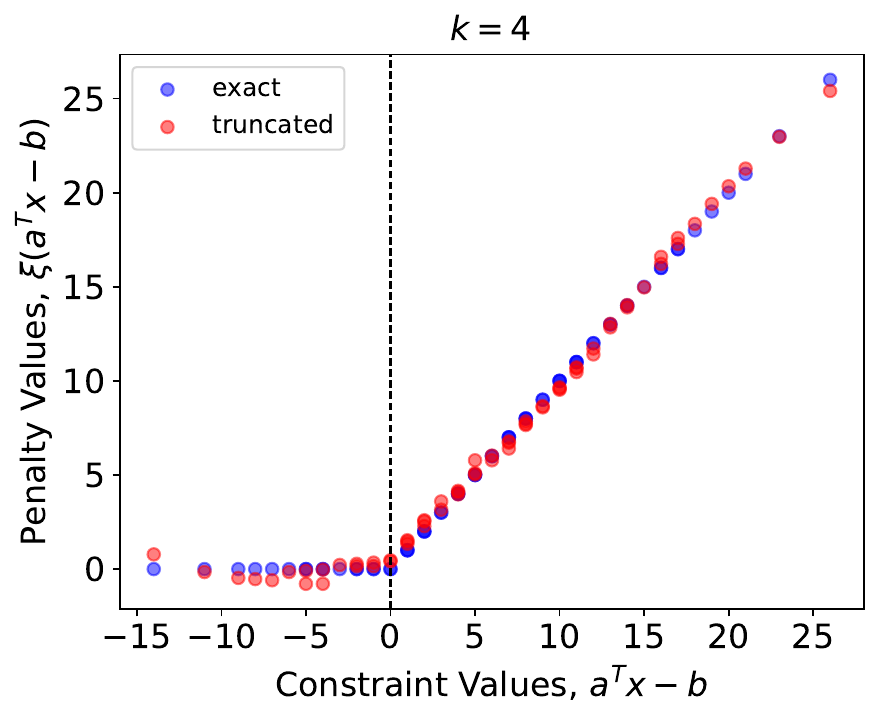}%
    }%
    \hfill
    \subfloat[]{%
        \includegraphics[width=0.32\textwidth]{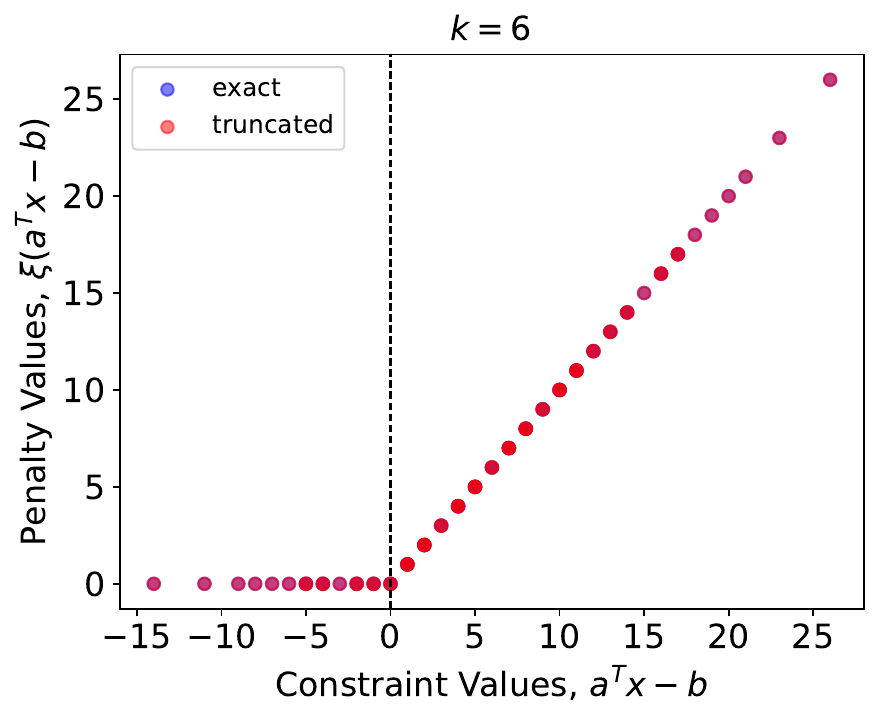}%
    }%
    \caption{Comparison of the penalty values given by the truncated Walsh spectrum $T_k\phi$ with the original penalty function $\phi$.
    First, second and third column correspond to $k=2,4,6$ respectively. (a)--(c) Heaviside step penalty in Eq.~\eqref{eq:heaviside}; (d)--(f) Hyperbolic tangent penalty
    $\frac{1}{2}(\tanh(y)+1)$; (g)--(i) Exponential penalty $\exp(0.1y)$; (j)--(l) ReLU in Eq.~\eqref{eq:relu}. The vertical dashed line shows $\ba^\top\bx -b=0$ the boundary between
    feasible and infeasible constraint values.}
    \label{fig:walsh-approx}
\end{figure*}

\begin{figure}
    \centering
    \includegraphics[width=0.95\linewidth]{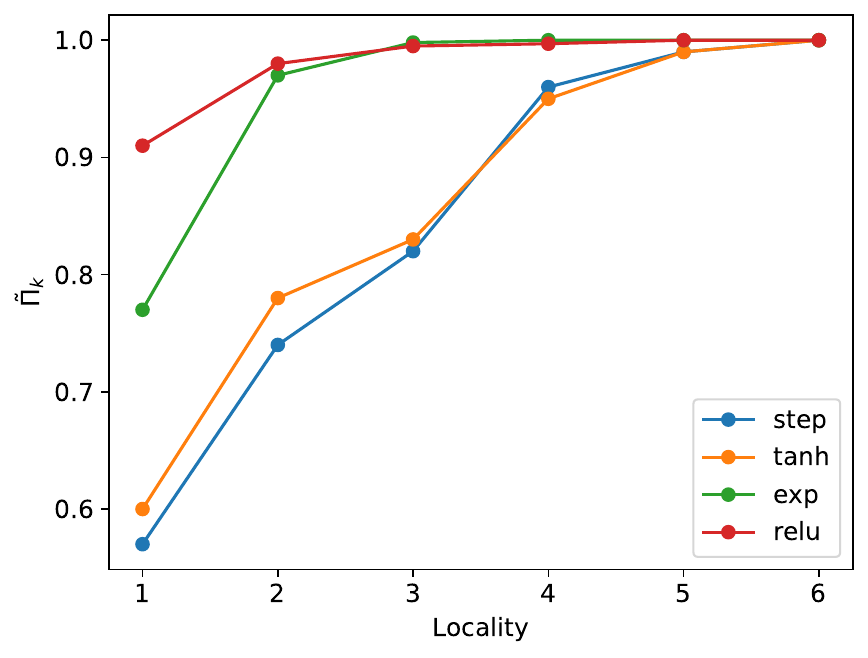}
    \caption{The fraction of retained power $\Pi_k$ of the truncated Walsh series against the truncation degree $k$ Eq.~\eqref{eq:retained-power}.
    Different functions require different degree to achieve the full spectral power. The exponential and ReLU achieved full spectral power at $k=3$,
    whereas Heaviside and hyperbolic tangent only achieved full spectral power at $k=6$ (the full spectrum).}
    \label{fig:pi-k}
\end{figure}

Next, we solve the Knapsack problem using WCP-QAOA, with the statevector simulation to simulate the QAOA circuit.
We choose the Heaviside step function as the penalty function. 
The solution bit-string is obtained as the bit-string with highest probability in the output statevector from the trained QAOA circuit.
Fig.~\ref{fig:qaoa-kp6} shows the results in box plots, where each box contains the results for 10 random initializations.
Several metrics are investigated against the truncation locality $k$ of the Walsh spectrum and the QAOA circuit depth $p$.
Fig.~\ref{fig:qaoa-kp6}(a) shows the gap between the QAOA solution objective and the ground truth (GT), defined as
\begin{equation}
    \Delta_{o}^{\text{GT}} = 1 - \frac{C_\text{QAOA}}{C_\text{GT}},
    \label{eq:gt-obj-gap}
\end{equation}
where $C$ is the objective \eqref{eq:kp-obj} evaluated at the solution.
Fig.~\ref{fig:qaoa-kp6}(b) shows the probability of sampling the GT solution $p_{\text{GT}}$.
Fig.~\ref{fig:qaoa-kp6}(c) shows the GT loss gap, the form is similar to Eq.~\eqref{eq:gt-obj-gap}, but the $C$ are replaced with losses:
\begin{equation}
    \Delta_{l}^{\text{GT}} = 1 - \frac{\mathcal{L}_\text{QAOA}}{\mathcal{L}_\text{GT}}.
    \label{eq:gt-loss-gap}
\end{equation}
Fig.~\ref{fig:qaoa-kp6}(d) shows the loss gap, which is defined as
\begin{equation}
    \Delta_{l}^{\min} = 1 - \frac{\mathcal{L}_\text{QAOA}}{\mathcal{L}_{\min}}.
    \label{eq:loss-gap}
\end{equation}
The distinction between Eq.~\eqref{eq:gt-loss-gap} and \eqref{eq:loss-gap} is required as the truncated Walsh series with $k < n$ might distort the Hamiltonian spectrum,
causing the solution with the lowest eigenvalue in the Hamiltonian with the truncated Walsh series $H_C^{(k)}$ to differ from that in $H_C$,
creating a false minimum in $H_C^{(k)}$.
The objective is nearer to the GT when $\Delta_o^\text{GT}$ is nearer to zero. Negative value of $\Delta_o^{\text{GT}}$ implies an infeasible solution that
has objective greater than the GT. This applies to other gaps as well.

From Fig.~\ref{fig:qaoa-kp6}(a), at $p=1$, $\Delta_{o}^{\text{GT}}$ have a large variance, and the variance decreases as the depth increases.
This shows that shallow QAOA is not expressive enough to find the solution, where it improves when the QAOA circuit becomes deeper.
QAOA also performs better with increased truncation locality $k$, which can be observed for deeper QAOA, the GT objective gap gets more concentrated to zero as $k$ increases.
Fig.~\ref{fig:qaoa-kp6}(b) shows this clearly with an increased probability of sampling the GT with increased $k$.
However, the probability only increases drastically from $k=6$, whereas it stays low in other truncation degrees, hinting that false minima continue to persist until $k=5$.
Fig.~\ref{fig:qaoa-kp6}(c) and (d) show that both $\Delta_{l}^{\text{GT}}$ and $\Delta_{l}^{\min}$ decrease as depth increases.
However, $\Delta_{l}^{\text{GT}}$ for $k=2$ is very near to zero, which is counterintuitive. This itself does not say that the QAOA is successful at $k=2$,
since both $\Delta_{o}^{\text{GT}}$ and $p_{\text{GT}}$ are bad at $k=2$. 

\begin{figure*}
    \centering
    \subfloat[]{%
        \includegraphics[width=0.48\textwidth]{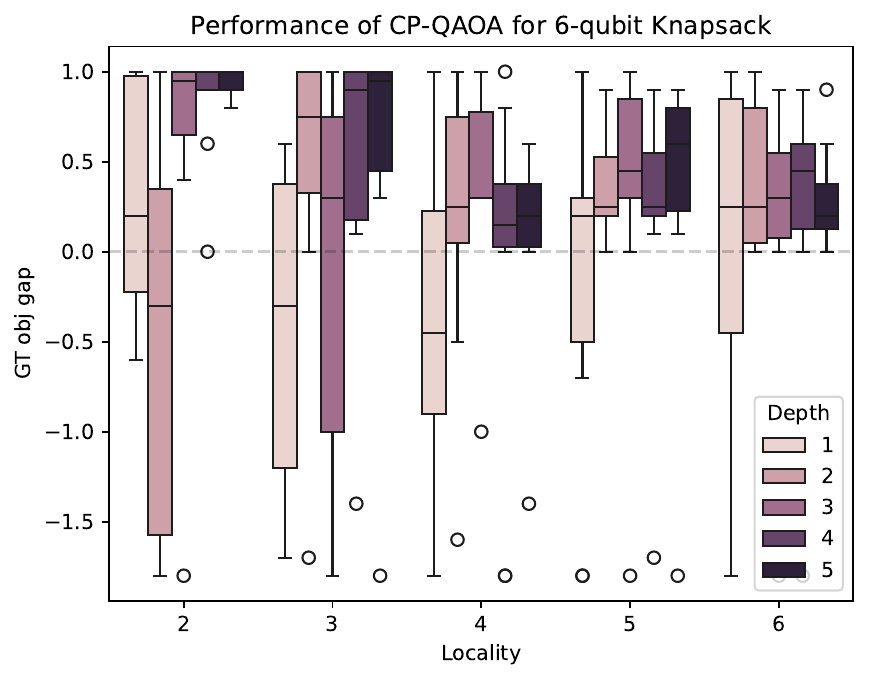}%
    }%
    \hfill
    \subfloat[]{%
        \includegraphics[width=0.48\textwidth]{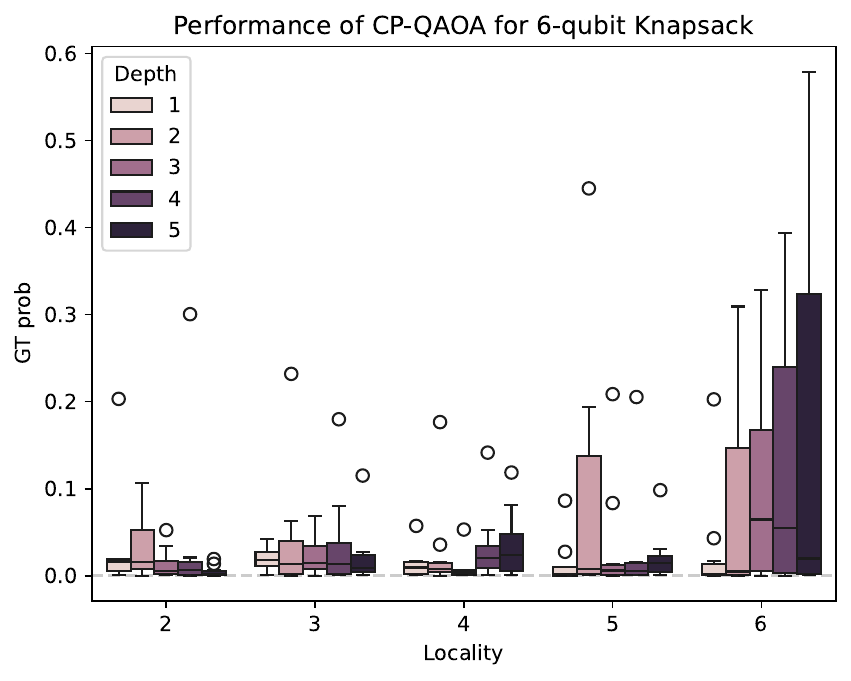}%
    }%
    \hfill
    \subfloat[]{%
        \includegraphics[width=0.48\textwidth]{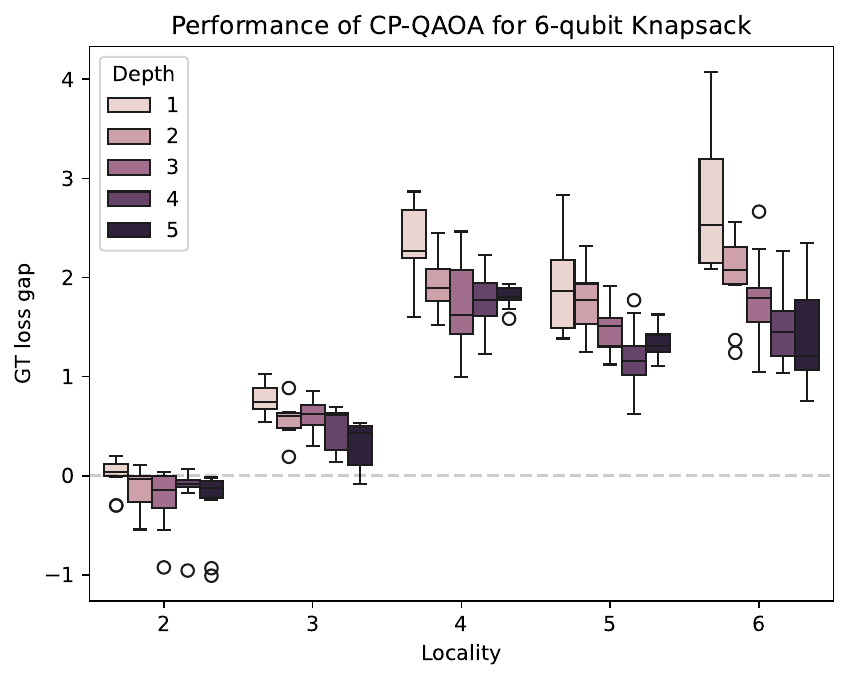}%
    }%
    \hfill
    \subfloat[]{%
        \includegraphics[width=0.48\textwidth]{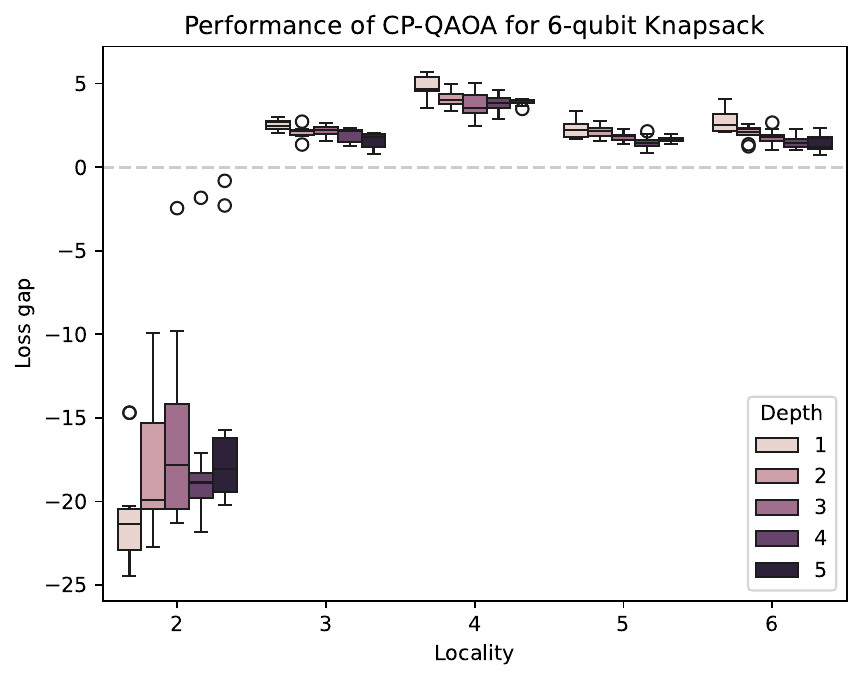}%
    }%
    \caption{Evaluation metrics for a 6-qubit Knapsack problem. WCP-QAOA is tested at different localities $k$ (truncation) and QAOA depths.
    We use the bit-string with the highest probability as the solution from QAOA. Each box shows the result for 10 random initializations. 
    The gaps have values with 0 closest to the target and larger magnitudes further to the target.
    (a) The objective gap with the ground truth (GT). (b) The probability of sampling the GT. (c) The gap of the final loss value with the GT.
    (d) The gap of the final loss value with the lowest eigenvalue in the Hamiltonian. The lowest eigenvalue of the Hamiltonian is not necessarily (in fact, most of the cases)
    the GT for $k < n$.}
    \label{fig:qaoa-kp6}
\end{figure*}

\section{Discussions}
\label{sec:discussions}

\subsection{Truncated Walsh spectrum is optimal in $L^2$}
\label{sec:walsh-opt}

We prove that the degree-$k$ truncation $T_k\phi$ is the best approximation for $\phi$, i.e. $T_k\phi$ has the least square error with respect to the original function $\phi$ 
in the multilinear space up to degree-$k$.

\begin{definition}[Inner product of functions]
    Let $u,v: \{0,1\}^n \rightarrow \mathbb{R}$, the inner product $\inner{u}{v}$ is defined as
    \begin{equation}
        \inner{u}{v} = \frac{1}{2^n}\sum_{\bx\in\{0,1\}^n} u(\bx) v(\bx).
    \end{equation}
    \label{def:inner-func}
\end{definition}

\begin{lemma}[Orthonormality of Walsh characters]
    The Walsh characters $\chi_S$ form an orthonormal basis for the vector space $V$ of the functions $\{0,1\}^n \rightarrow \mathbb{R}$. Therefore,
    \begin{equation}
        \inner{\chi_S}{\chi_T} = 
        \begin{cases}
            1 & \text{if } S = T \\
            0 & \text{if } S \neq T.
        \end{cases}
    \end{equation}
    \label{lem:inner-chi}
\end{lemma}

\begin{lemma}[$T_k\phi$ is an orthogonal projection]
    Let $\phi: \{0,1\}^n \rightarrow \mathbb{R}$ and $V_k = \Span\{\chi_S: |S|\leq k\}$. $T_k\phi$ is an orthogonal projection to $\phi$, i.e.
    \begin{equation}
        \inner{\phi-T_k\phi}{\chi_T} = 0,
    \end{equation}
    where $T$ satisfies $|T|\leq k$.
    \label{lem:tk-orth}
\end{lemma}
\begin{proof}
    We want to prove that $\phi-T_k\phi$ is orthogonal to the bases of $V_k$. 
    \begin{align}
        \inner{\phi-T_k\phi}{\chi_T} & = \inner{\phi}{\chi_T} - \inner{T_k\phi}{\chi_T} \\
        & = \inner{\phi}{\chi_T} - \inner{\sum_{|S|\leq k} \hat{\phi}(S) \chi_S}{\chi_T} \\ 
        & = \inner{\phi}{\chi_T} - \sum_{|S|\leq k} \hat{\phi}(S) \inner{\chi_S}{\chi_T}.
    \end{align}
    Using Definition~\ref{def:inner-func} and Eq.~\eqref{eq:coeff} for the first term, and Lemma~\ref{lem:inner-chi} for the second term gives
    \begin{equation}
        \inner{\phi-T_k\phi}{\chi_T} = \hat\phi(T) - \hat\phi(T) = 0.
    \end{equation}
\end{proof}

\begin{lemma}[Pythagorean theorem for Hilbert space]
    Any pair of orthogonal vectors $\inner{u}{v} = 0$ yields
    \begin{equation}
        \twonorm{u + v}^2 = \twonorm{u}^2 + \twonorm{v}^2 .
    \end{equation}
    Therefore, 
    \begin{align}
        \phi & = (\phi - T_k\phi) + T_k\phi \\
        \twonorm{\phi}^2 & = \twonorm{\phi - T_k\phi}^2 + \twonorm{T_k\phi}^2. 
    \end{align}
    \label{lem:pythagorean}
\end{lemma}
\begin{proof}
    Suppose $u,v$ satisfies $\inner{u}{v} = 0$, then 
    \begin{align}
        \twonorm{u + v}^2 & = \inner{u+v}{u+v} \\
        & = \inner{u}{u} + \inner{u}{v} + \inner{v}{u} + \inner{v}{v} \\ 
        & = \twonorm{u}^2 + \twonorm{v}^2.
    \end{align}
\end{proof}

\begin{theorem}[$T_k\phi$ is the best approximation for $\phi$]
    For any arbitrary function $u\in V_k$, its $L^2$-norm difference between $\phi$ is lower bounded by the difference between $\phi$ and $T_k\phi$:
    \begin{equation}
        \twonorm{\phi-u} \geq  \twonorm{\phi-T_k\phi},
    \end{equation}
    which means that $T_k\phi$ has the least square difference with $\phi$ in the vector space.
    \label{thm:walsh-opt}
\end{theorem}
\begin{proof}
    Let $u\in V_k$. Since $T_k\phi\in V_k$, we have $T_k\phi - u\in V_k$. Since $\phi-T_k\phi \perp V_k$, we have $\phi-T_k\phi \perp T_k\phi -u$.
    From Lemma~\ref{lem:pythagorean},
    \begin{align}
        \phi - u & = (\phi - T_k\phi) + (T_k\phi - u) \\
        \twonorm{\phi - u}^2 & = \twonorm{\phi - T_k\phi}^2 + \twonorm{T_k\phi - u}^2.
    \end{align}
    Since all the terms are positive, we conclude that
    \begin{equation}
        \twonorm{\phi - u} \geq \twonorm{\phi-T_k\phi}.
    \end{equation}
\end{proof}

In fact, we can also show that any orthogonal projection of $\phi$ is the best approximation of it. We wrote Theorem~\ref{thm:walsh-opt} in a way that suits our context of 
the truncated Walsh series.

The consequence of Theorem~\ref{thm:walsh-opt} is that one cannot get a better approximation for $\phi$ than the truncated Walsh series $T_k\phi$, since it already has the least square error.
This also means that the approximation of a function via the truncation of Taylor series, as used in the unbalanced penalization~\cite{ub-penalty}, is only equal or worse than the 
truncated Walsh series, which we show later in the section.

\subsection{Closed-form spectrum for exponential penalty and the relation to unbalanced penalization}
\label{sec:exp-spectrum}

\begin{proposition}[Closed-form Walsh spectrum for exponential penalty]
\label{prop:closed-exp}
    The exponential penalty can be expressed as the product of 1-local spins:
    \begin{align}
       \phi_\eta(\bx) & = e^{\eta(\ba^\top\bx-b)}\\
       \phi_\eta(\bz) & = e^{-\eta b}\prod_i (\alpha_i + \delta_iz_i),
       \label{eq:exp-closed}
    \end{align}
    where 
    \begin{align}
        \alpha_i & = \tfrac{1}{2}\big(1 + e^{\eta a_i}\big) \label{eq:alpha} \\
        \delta_i & = \tfrac{1}{2}\big(1 - e^{\eta a_i}\big). \label{eq:delta}
    \end{align}
    In terms of Pauli-$Z$,
    \begin{equation}
        e^{-\eta b} \prod_i (\alpha_i + \delta_i Z_i) = e^{-\eta b} \bigotimes_i
        \begin{pmatrix}
            1 & 0 \\
            0 & e^{\eta a_i}
        \end{pmatrix}.
    \end{equation}
    In fact, the expansion of the product is not required because the Walsh coefficients have a closed-form:
    \begin{equation}
        \hat{\phi}_\eta(S) = e^{-\eta b}\prod_{i\in S}\delta_i \prod_{i\notin S}\alpha_i,
        \label{eq:exp-coeffs}
    \end{equation}
    Equivalently, we can write $t_i = \tanh(\eta a_i/2)$, then
    \begin{equation}
        \hat{\phi}_\eta(S) = \hat{\phi}_\eta(\varnothing)\prod_{i\in S}\left(-t_i\right),
        \label{eq:tanhform}
    \end{equation}
    where
    \begin{equation}
        \hat{\phi}_\eta(\varnothing) = e^{\eta(\frac{1}{2}\onenorm{\ba} - b)}\prod_{i=1}^{n}\cosh\frac{\eta a_i}{2}.
        \label{eq:tanhform-const}
    \end{equation}
\end{proposition}

The proof for the above is shown in Appendix~\ref{app:exp-product}.
This shows that for the exponential penalty, no transform has to be computed at all.
Every coefficient is a product of $n$ known numbers, obtainable in $\BigO(|S|)$ time once $\hat{\phi}_\eta(\varnothing)$ and the $t_i$ are cached.
It also exposes the structure that governs truncation. Since $0 \le t_i < 1$, the magnitudes decay multiplicatively as $S$ grows,
and the decay rate of qubit $i$ is set by $\eta a_i$: a variable with a larger constraint coefficient contributes more to the high-degree terms.

Although the exponential penalty has nice properties, its caveat is also obvious. 
Unlike the Heaviside step function which has a clear jump of value between feasible and infeasible, scaled by the penalty factor,
the exponential penalty values are relatively small at small values of infeasible regime $h(\bx) > 0$, which might not have enough separation with the feasible regime.
Then, the penalty value skyrockets as the constraint value increases, which might cause numerical issues. 
Hence, the exponential penalty must be carefully designed during implementation. 

Unbalanced penalization~\cite{ub-penalty} implements the exponential penalty $\phi_\eta$ with
\begin{equation}
    U^{(2)}_\eta(\bx) = 1 + \eta h(\bx) + \tfrac{1}{2}\eta^2 h(\bx)^2 ,
    \label{eq:unbalanced}
\end{equation}
which is the degree-2 Taylor truncation. Hence, we write
\begin{equation}
    U^{(k)}_\eta(\bx) = \sum_{m=0}^{k}\frac{\big(\eta h(\bx)\big)^{m}}{m!}
    \label{eq:taylor-k}
\end{equation}
for the order-$k$ Taylor truncation.

\begin{corollary}[Taylor truncation is lower-bounded by Walsh truncation]
    For any Taylor truncation with truncation degree $k \geq 0$, its error with the original function is lower bounded by the degree-$k$ Walsh truncation:
    \begin{equation}
        \twonorm{\phi_\eta - T_k\phi_\eta} \leq \twonorm{\phi_\eta - U_\eta^{(k)}}.
    \end{equation}
\end{corollary}
\begin{proof}
    Since $U_\eta^{(k)}$ is multilinear up to degree-$k$, $U_\eta^{(k)} \in V_k$.
    Hence, the lower bound can be derived directly from Theorem~\ref{thm:walsh-opt} by replacing the function $u$ with the truncated Taylor $U_\eta^{(k)}$.
\end{proof}
This shows that Walsh truncation is never worse than the exponential penalty compared to Taylor truncation, with the same degree of truncation $k$.
There are several advantages in choosing the Walsh penalty. First, the penalty has to be differentiable at $h(\bx)=0$ to be expressed with the Taylor expansion,
which clearly does not apply to functions like the step function or ReLU. 
Second, the Walsh series exactly reproduces the original penalty with finite terms at $k=n$, in which the Taylor series will never achieve because the series is infinite.
Third, the Walsh coefficients for the exponential penalty are easy to compute, as stated in Proposition~\ref{prop:closed-exp}.

\subsection{Structural shortcuts}
\label{sec:structure}

We discuss some of the structural shortcuts that will further reduce the computational cost of WCP-QAOA.
These shortcuts are usually caused by the structures that are inherent to the problem. 

\begin{proposition}[Sparse constraints]
\label{prop:sparse}
Let $R = \supp(\ba)$, the number of nonzero elements in $\ba$, with $|R| = r$.
Then $\hat\phi(S) = 0$ unless $S\subseteq R$, so at most $2^r$ coefficients are nonzero and $H_\phi$ acts on $r$ qubits.
\end{proposition}

\begin{proof}
$\phi$ does not depend on $x_i$ for $i\notin R$; averaging Eq.~\eqref{eq:coeff} over that coordinate gives $\hat\phi(S)=0$ whenever $i \in S$.
\end{proof}

\begin{proposition}[Cardinality constraints]
\label{prop:cardinality}
If all nonzero $a_i$ are equal, $\phi$ is a symmetric function of the corresponding bits, and $\hat\phi(S)$ depends on $S$ only through $|S|$.
There are therefore at most $n+1$ distinct coefficient values, and the phase separator uses only $n+1$ distinct rotation angles.
\end{proposition}

\begin{proof}
Permuting the bits in $\supp(\ba)$ leaves $\phi$ invariant and maps $\chi_S\mapsto\chi_{\pi(S)}$. Eq.~\eqref{eq:coeff} then gives $\hat\phi(S) = \hat\phi(\pi(S))$
for every permutation $\pi$ of $\supp(\ba)$.
\end{proof}

\begin{proposition}[Balanced thresholds annihilate even degrees]
\label{prop:parity}
Suppose $\onenorm{\ba}$ is odd and $b = \onenorm{\ba}/2$, and suppose $\xi$ satisfies $\xi(y)+\xi(-y) = c$ for all $y\ne 0$, as the Heaviside step function does with $c=1$.
Then $\hat\phi(S) = 0$ for every $S$ with $|S|$ even and $|S|\ge 2$.
\end{proposition}

\begin{proof}
Let $\bbx$ be the bitwise complement of $\bx$. Since $\ba^{\top}\bbx = \onenorm{\ba} - \ba^{\top}\bx$, the choice $b = \onenorm{\ba}/2$ gives $h(\bar\bx) = -h(\bx)$.
Because $\onenorm{\ba}$ is odd, $b$ is a half-integer and $h(\bx)\ne 0$ for every $\bx$, so $\phi(\bx)+\phi(\bbx) = c$ identically.
Set $\varphi = \phi - c/2$, which obeys $\varphi(\bbx) = -\varphi(\bx)$.
Using $\chi_S(\bbx) = (-1)^{|S|}\chi_S(\bx)$ and reindexing the sum in Eq.~\eqref{eq:coeff} by $\bx\mapsto\bbx$ gives $\hat\varphi(S) = (-1)^{|S|+1}\hat\varphi(S)$,
so $\hat\varphi(S) = 0$ for even $|S|$. Finally $\hat\phi(S) = \hat\varphi(S)$ for all $S\ne\varnothing$.
\end{proof}

Proposition~\ref{prop:parity} halves the term count in the balanced regime: a degree-$3$ truncation then retains only weight-$1$ and weight-$3$ terms,
and a degree-$4$ truncation buys nothing over a degree-$3$ one. It also makes the general point concrete. Moving $b$ away from $\onenorm{\ba}/2$ restores the even-degree terms, 
so the same penalty profile can be cheap or expensive depending on where the threshold sits relative to the constraint coefficients.
This is one reason why the interplay between $\ba$, $b$ and $\xi$ deserves study in its own right, and the three propositions above are best read as
the cases in which that interplay is favorable enough to identify in advance.

\section{Conclusion}
\label{sec:conclusion}

We have shown that slack-free custom penalties, which in our previous work required a problem-agnostic VQE ansatz, can be implemented exactly in the phase separator of QAOA.
The key observation is that the Walsh--Hadamard expansion of a penalty $\xi(\ba^{\top}\bx-b)$ is its Pauli-$Z$ expansion (Lemma~\ref{lem:dictionary}).
Consequently, for any penalty profile $\xi$, the phase separator of WCP-QAOA factorizes into commuting $\mathrm{CNOT}$--$R_z$ blocks with no ancilla qubit and no Trotter error.
Truncating the Walsh series at degree $k$ keeps $\BigO(n^k)$ Pauli terms on the original $n$ qubits.
The truncation is the best degree-$k$ approximation of the penalty in $L^2$ (Theorem~\ref{thm:walsh-opt}), so it is never worse in $L^2$ than any polynomial surrogate of the same degree, including the Taylor surrogate that underlies unbalanced penalization~\cite{ub-penalty}.
Proposition~\ref{prop:separation} gives a sufficient condition on the discarded Walsh weight under which every feasible bit-string remains below every infeasible one.

The Walsh coefficients do not require an enumeration of the $2^n$ bit-strings.
Theorem~\ref{thm:sub-exp-coeff} reduces each coefficient to $n$ sparse convolutions, which gives any coefficient exactly in $\BigO(n\onenorm{\ba})$ time, and hence the degree-$k$ Hamiltonian and its retained spectral power in time polynomial in $n$ and $\onenorm{\ba}$ for fixed $k$.
For the exponential penalty no transform is needed at all: its spectrum factorizes over the qubits, so every coefficient follows in closed form from the constraint coefficients (Proposition~\ref{prop:closed-exp}).
The structure of the constraint simplifies the exact spectrum before any truncation is applied: a sparse support limits the number of nonzero coefficients, equal constraint coefficients limit the number of distinct ones, and a balanced threshold removes every even-degree term beyond the constant for step-like penalties (Proposition~\ref{prop:parity}).

The numerical study of a 6-qubit knapsack instance shows that the cost of truncation strongly depends on the penalty profile.
The exponential and ReLU penalties retain almost all of their spectral power by degree three, whereas the step and the hyperbolic-tangent penalties spread it up to the full degree.
Consistently, the probability of sampling the optimum with WCP-QAOA at depths up to 5 increases appreciably only when the full spectrum is retained.
Two lessons follow.
First, $L^2$ optimality under the uniform measure, and hence a large retained power $\Pi_k$, does not guarantee that the ground state is preserved, because the $L^2$ criterion weighs all $2^n$ bit-strings equally, whereas the outcome of the optimization is decided by the ordering of the lowest-energy states.
Second, the large penalty factor that separates feasible from infeasible states also amplifies the distortion that the truncation introduces within the feasible set, so feasibility and optimality pull the penalty factor in opposite directions.

There are several directions for future work.
A computable criterion that certifies the preservation of the low-energy tail, rather than the full feasible--infeasible ordering of Proposition~\ref{prop:separation}, would make the choice of $k$ and $\lambda$ principled, particularly in combination with CVaR training~\cite{cvar,cvar-cp}.
Penalty profiles could be designed to concentrate their spectrum at low degree for a prescribed penalty gap, and truncations could be fitted under measures that emphasize the low-energy region rather than weighting all bit-strings equally.
On the circuit side, CNOT-efficient synthesis of sparse truncated spectra and the combination with feasibility-preserving mixers~\cite{xy-mixer,grover-mixer} remain to be explored.
The parameter landscape of WCP-QAOA also deserves study: since the Walsh coefficients are generally not integers, the $2\pi$-periodicity in $\gamma$ enjoyed by integer-valued cost Hamiltonians is not guaranteed, and it is open which of the regularities behind QAOA initialization heuristics carry over.
Finally, systematic benchmarks against slack-based QUBO and unbalanced penalization on larger and multiply constrained instances, such as the MDKP, are needed to establish when the Walsh route offers a practical advantage.

\begin{acknowledgments}
This project is supported by the National Research Foundation, Singapore through the National Quantum Office under its Quantum Engineering Programme (QEP) 3.0 (Project ID: W24Q3D0002). Any opinions, findings, and conclusions, or recommendations expressed in this material are those of the author(s) and do not reflect the views of National Research Foundation, Singapore. 
\end{acknowledgments}

\appendix

\section{Proof of Theorem~\ref{thm:sub-exp-coeff}}
\label{app:convolution}

Let $\zeta$ be a formal variable, used purely as a bookkeeping label: the coefficient of $\zeta^{s}$ records data attached to the constraint value $s$, and $\zeta$ is never assigned a numerical value. For each bit $i$ define
\begin{align}
    g_i(\zeta) & = \sum_{x_i\in\{0,1\}} \frac{1}{2}(-1)^{x_i[i\in S]}\zeta^{a_ix_i}
    \label{eq:app-gi-def}\\
    &= \frac{1}{2} + \frac{1}{2}(-1)^{[i\in S]}\zeta^{a_i},
    \label{eq:app-gi}
\end{align}
The two terms of Eq.~\eqref{eq:app-gi-def} are the two possible values of bit $i$: choosing $x_i = 1$ advances the exponent by $a_i$, which is that bit's contribution to $s = \ba^{\top}\bx$, and multiplies the sign by $-1$ if $i\in S$, which is that bit's contribution to the parity $\chi_S(\bx)$.

Define
\begin{equation}
    G_S(\zeta) = \prod_{i=1}^{n}g_i(\zeta)
    = \prod_{i=1}^{n}\ \sum_{x_i\in\{0,1\}}\frac{1}{2}(-1)^{x_i[i\in S]}\zeta^{a_ix_i}.
    \label{eq:app-GS}
\end{equation}
Expanding a product of sums into a sum over all combinations of choices is the distributive law, and a choice of one term from each factor is exactly a bit-string $\bx\in\{0,1\}^n$. Hence
\begin{align}
    G_S(\zeta)
    &= \sum_{\bx\in\{0,1\}^n}\ \prod_{i=1}^{n}\frac{1}{2}(-1)^{x_i[i\in S]}\zeta^{a_ix_i}
    \label{eq:app-expand}\\
    &= \frac{1}{2^n}\sum_{\bx\in\{0,1\}^n}(-1)^{\sum_{i\in S}x_i}\ \zeta^{\ba^{\top}\bx}
    \label{eq:app-collect}\\
    &= \frac{1}{2^n}\sum_{\bx\in\{0,1\}^n}\chi_S(\bx)\zeta^{\ba^{\top}\bx}.
    \label{eq:app-chi}
\end{align}
Every bit-string contributes a single monomial $\zeta^{\ba^{\top}\bx}$, so grouping the sum by the value $s = \ba^{\top}\bx$ gives
\begin{equation}
    G_S(\zeta) = \sum_{s=0}^{\onenorm{\ba}}
    \underbrace{\Bigg(\frac{1}{2^n}\sum_{\bx:\,\ba^{\top}\bx = s}\chi_S(\bx)\Bigg)}_{\textstyle =\ \omega_S(s)}\zeta^{s}.
    \label{eq:app-group}
\end{equation}
Thus, $\bm\omega_S$ is the coefficient vector of $G_S$. Since the coefficient vector of a product of polynomials is the convolution of their coefficient vectors,
Eq.~\eqref{eq:app-GS} gives
\begin{equation}
    \bm\omega_S = \bm{g}_1 * \bm{g}_2 * \cdots * \bm{g}_n,
    \label{eq:app-conv}
\end{equation}
which is Eq.~\eqref{eq:convolution}. Explicitly, $\bm{g}_i$ is the vector with
\begin{equation}
    \bm{g}_i[s] =
    \begin{cases}
        \phantom{-}\tfrac{1}{2} & s = 0,\\[2pt]
        -\tfrac{1}{2} & s = a_i \ \text{and}\ i\in S,\\[2pt]
        \phantom{-}\tfrac{1}{2} & s = a_i \ \text{and}\ i\notin S,\\[2pt]
        \phantom{-}0 & \text{otherwise},
    \end{cases}
    \label{eq:app-gvec}
\end{equation}
with the two cases $s=0$ and $s=a_i$ added together when $a_i = 0$.

For the complexity, each $\bm{g}_i$ has at most two nonzero entries, so convolving it with a running vector truncated to length $\onenorm{\ba}+1$ costs $O(\onenorm{\ba})$.
Performing $n$ such convolutions costs $O(n\onenorm{\ba})$, and the final evaluation $\hat\phi(S) = \sum_s \xi(s-b)\omega_S(s)$ costs a further $O(\onenorm{\ba})$.
Only two vectors of length $\onenorm{\ba}+1$ need be stored. \hfill$\square$

\section{Proof for Proposition~\ref{prop:closed-exp}}
\label{app:exp-product}

For the function closed-form Eq.~\eqref{eq:exp-closed}, the exponential penalty can be decomposed into
\begin{equation}
    e^{\eta(\ba^\top\bx-b)}=e^{-\eta b}\prod_i e^{\eta a_ix_i}.
\end{equation}
Since each factor is $e^{\eta a_i(1-z_i)/2}$ in the spin domain $z_i\in\{1,-1\}$, matching the values when $z_i=1\rightarrow \phi_\eta = 1$ and $z_i=-1\rightarrow \phi_\eta = e^{\eta a_i}$
gives a function linear in $z_i$. 
\begin{equation}
    e^{\eta a_i(1-z_i)/2} = \alpha_i + \delta_i z_i,
\end{equation}

For the closed-form of Walsh coefficients Eq.~\eqref{eq:exp-coeffs}, since $\phi_\eta(\bx ) = e^{\eta(\ba^\top \bx - b)}$,

\begin{align*}
    \hat{\phi}_\eta(S) &= \frac1{2^n}\sum_{x \in \{0,1\}^n} e^{-\eta b} \prod_{i=1}^n e^{\eta a_ix_i}\chi_S(x)\\
    &=\frac1{2^n}e^{-\eta b}\sum_{x \in \{0,1\}^n}  \left(\prod_{i=1}^n e^{\eta a_ix_i}\right)\left(\prod_{i \in S}(-1)^{x_i}\right)\\
    &=\frac1{2^n}e^{-\eta b}\sum_{x \in \{0,1\}^n}  \left(\prod_{i\in S} e^{\eta a_ix_i}(-1)^{x_i}\right)\left(\prod_{i \notin S}e^{\eta a_i x_i}\right)
\end{align*}

Now, we set $z_i = (-1)^{x_i}$ and express $e^{\eta a_ix_i}=\alpha_i + \delta_i z_i.$ This would simplify the expression as follows:

\begin{align*}
&\hat{\phi}_\eta(S)\\
    &= \frac1{2^n}e^{-\eta b}\sum_{z \in \{-1,1\}^n}  \left(\prod_{i\in S} (\alpha_i + \delta_i z_i)z_i\right)\left(\prod_{i \notin S}(\alpha_i + \delta_iz_i)\right)\\
    &=\frac1{2^n}e^{-\eta b}\sum_{z \in \{-1,1\}^n}  \left(\prod_{i\in S} (\alpha_iz_i + \delta_i )\right)\left(\prod_{i \notin S}(\alpha_i + \delta_iz_i)\right)
    \\&=\frac1{2^n}e^{-\eta b}\prod_{i \in S} \left(\sum_{z_i\in \{-1,1\}} (\alpha_iz_i + \delta_i )\right)\prod_{i \notin S}\left(\sum_{z_i \in \{-1,1\}}(\alpha_i + \delta_iz_i)\right)
        \\&=\frac1{2^n}e^{-\eta b}\prod_{i \in S}  (2\delta_i )\prod_{i \notin S}\left(2\alpha_i \right)
    \\&=e^{-\eta b}\prod_{i\in S}\delta_i \prod_{i\notin S}\alpha_i,
\end{align*}
\hfill$\square$


\bibliography{qaoaref,ccop}

\end{document}